\documentclass[11pt,a4paper]{article}
\usepackage[utf8]{inputenc}
\usepackage{amsmath,amsfonts,amssymb,amsthm,mathtools}
\usepackage{authblk}
\usepackage{latexsym,graphicx}
\usepackage{dsfont}
\usepackage{amscd}
\usepackage[dvipsnames]{xcolor}
\usepackage{extarrows} 
\numberwithin{equation}{section} 
\usepackage{yfonts}
\usepackage{authblk}
\usepackage{appendix}
\usepackage[colorlinks=true]{hyperref}
\hypersetup{urlcolor=blue, citecolor=red, linkcolor=blue}

\usepackage{tikz}
	 \usetikzlibrary{calc} 
	\usetikzlibrary{matrix,arrows,decorations.pathmorphing} 
\usepackage{tikz-3dplot}

\allowdisplaybreaks

\newtheorem{theorem}{Theorem}[section]
\newtheorem{corollary}[theorem]{Corollary}
\newtheorem{lemma}{Lemma}[section]
\newtheorem{proposition}{Proposition}[section]
\theoremstyle{remark}
\newtheorem{remark}{Remark}[section]

\def\Xint#1{\mathchoice
{\XXint\displaystyle\textstyle{#1}}%
{\XXint\textstyle\scriptstyle{#1}}%
{\XXint\scriptstyle\scriptscriptstyle{#1}}%
{\XXint\scriptscriptstyle\scriptscriptstyle{#1}}%
\!\int}
\def\XXint#1#2#3{{\setbox0=\hbox{$#1{#2#3}{\int}$}
\vcenter{\hbox{$#2#3$}}\kern-.5\wd0}}

\def\pvint{\Xint-}

\newcommand{\im}{\mathrm{Im} }
\newcommand{\aR}{a^\mathrm{R}}

\newcommand{\aI}{a^{\mathrm{I}}}

\newcommand{\ii}{\mathrm{i}}
\newcommand{\dd}{\mathrm{d}}

\newcommand{\R}{{\mathbb R}}
\newcommand{\C}{{\mathbb C}}
\newcommand{\Z}{{\mathbb Z}}

\newcommand{\tr}{\mathrm{tr}}

\newcommand{\bzero}{\mathbf{0}}

\newcommand{\ve}{\mathbf{e}}
\newcommand{\vf}{\mathbf{f}}
\newcommand{\bm}{\mathbf{m}}
\newcommand{\bn}{\mathbf{n}}
\newcommand{\bsigma}{\boldsymbol{\sigma}}

\newcommand{\cA}{\mathcal{A}}

\newcommand{\cF}{\mathcal{F}}
\newcommand{\cG}{\mathcal{G}}
\newcommand{\cH}{\mathcal{H}} 
\newcommand{\cN}{\mathcal {N}}
\newcommand{\cT}{\mathcal{T}}
\newcommand{\cU}{\mathcal{U}}
\newcommand{\cV}{\mathcal{V}}

\newcommand{\tT}{\tilde{T}}

\newcommand{\mA}{\mathsf{A}}
\newcommand{\mB}{\mathsf{B}}
\newcommand{\mC}{\mathsf{C}}

\newcommand{\mP}{\mathsf{P}}
\newcommand{\mQ}{\mathsf{Q}}
\newcommand{\mU}{\mathsf{U}}
\newcommand{\mV}{\mathsf{V}}

\newcommand{\ocomma}{\,\overset{\circ}{,}\,}

\title{Spin generalizations of the Benjamin-Ono equation}

\author[1]{Bjorn K. Berntson}
\author[1,2]{Edwin Langmann}
\author[3]{Jonatan Lenells}
\date{January 18, 2022}

\affil[1]{Department of Physics, KTH Royal Institute of Technology, SE-106 91 Stockholm, Sweden}
\affil[2]{Nordita, KTH Royal Institute of Technology and Stockholm University, SE-106 91 Stockholm, Sweden}
\affil[3]{Department of Mathematics, KTH Royal Institute of Technology, SE-100 44 Stockholm, Sweden}

\begin{document}

\maketitle
  
\begin{abstract}
We present new soliton equations related to the $A$-type spin Calogero-Moser (CM) systems introduced by Gibbons and Hermsen. 
These equations are spin generalizations of the Benjamin-Ono (BO) equation and the recently introduced non-chiral intermediate long-wave (ncILW) equation. 
We obtain multi-soliton solutions of these spin generalizations of the BO equation and the ncILW equation via a spin-pole ansatz where the spin-pole dynamics is governed by the spin CM system  in the rational and hyperbolic cases, respectively. 
We also propose physics applications of the new equations,  and we introduce a spin generalization of the standard intermediate long-wave equation which interpolates between the matrix Korteweg-de Vries equation, the Heisenberg ferromagnet equation, and the spin BO equation. 
\end{abstract} 

\section{Introduction}
\label{sec:intro}
In this paper, we introduce and solve a new exactly solvable nonlinear partial integro-differential equation which is a natural spin generalization of the Benjamin-Ono (BO) equation \cite{benjamin1967,ono1975} and which we therefore call the {\em spin BO (sBO) equation}. We present arguments that the sBO equation is not only interesting from a mathematical point of view but also for physics. We also introduce, discuss, and present results about a closely related spin generalization of the non-chiral intermediate long-wave (ncILW) equation recently introduced and studied by us in \cite{berntson2020,berntsonlangmann2020,berntsonlangmann2021}.

The sBO equation describes the time evolution of a square matrix-valued function $\mU=\mU(x,t)$ depending on position and time variables $x\in\R$ and $t\in\R$, respectively, and it is given by 
\begin{equation} 
\label{eq:sBO} 
\boxed{ \mU_t + \{\mU,\mU_x\} +  H\mU_{xx}  +\ii [\mU,H\mU_x]=0 }
\end{equation} 
where $\mU_t$ is short for $\frac{\partial}{\partial t}\mU$ etc., $H$ is the usual Hilbert transform (which, for functions $f$ of $x\in\R$ is given by 
\begin{equation}
\label{eq:H}  
(Hf)(x)\coloneqq \frac1{\pi}\pvint_{\R} \frac1{x'-x} f(x')\,\dd{x'}
\end{equation} 
where $\pvint$ is the principal value integral), $\ii\coloneqq\sqrt{-1}$, and $[\cdot,\cdot]$ and $\{\cdot,\cdot\}$ denote the commutator and anti-commutator of square matrices, respectively (the matrix dimension $d\times d$ is arbitrary, with $d=1$ corresponding to the standard BO equation). 
 The sBO equation is interesting for several reasons:  (i) It is exactly solvable, and it includes well-known soliton equations as limiting cases; in particular, the sBO equation \eqref{eq:sBO} not only generalizes the BO equation but also the half-wave maps (HWM) equation \cite{zhou2015,lenzmann2018,lenzmann2018b}. 
(ii) It describes interesting physics beyond what has previously been described by exactly solvable equations; in particular, the sBO equation provides a hydrodynamic description of interacting particles with internal spin degrees of freedom, and it determines the evolution of the nonlinearly coupled charge- and spin-densities of this particle system. As will be explained in more detail, integrable systems of this kind are relevant for the quantum Hall effect in situations where the electron spin cannot be ignored, a topic which has received considerable interest in the physics literature in recent years;  see e.g.\ \cite{senthil1999,mccann2006,bernevig2006}. 
(iii) The sBO equation has multi-soliton solutions constructed via a simple spin-pole ansatz where the time evolution of the spins and poles is determined by the spin generalization  of the $A$-type Calogero-Moser (CM) system due to Gibbons and Hermsen \cite{gibbons1984}, which we call the spin CM (sCM) system; see also \cite{wojciechowski1987}. (iv) It provides another example of a correspondence between Calogero-Moser-Sutherland type systems and soliton equations; in particular, it is known that the BO equation and the ncILW equation are naturally related to $A$-type Calogero-Moser-Sutherland systems in several different ways \cite{berntson2020}, and while \eqref{eq:sBO} is a spin generalization of the BO equation, we also present a spin generalization of the ncILW equation and thus propose a correspondence between soliton equations and sCM systems in all four cases: rational (I), trigonometric (II), hyperbolic (III), elliptic (IV) (see, e.g., \cite{olshanetsky1981} for basic facts about CM systems).
We substantiate this proposal in cases I--III; case IV is technically more demanding and thus left for future work. 
(v) Our results suggest that this equation is exactly solvable in the same strong sense as the BO and HWM equations, and this opens up possibilities for several future research projects. 

In the following, we explain the above points (i)--(v) in more detail; we emphasize that these points are complementary and, depending on the interests of the reader, some of them can be ignored without loss of continuity. 

\paragraph{(i) Special cases.} 
The BO equation is an exactly solvable nonlinear partial integro-differential equation that describes one-dimensional internal waves in deep water. In particular, physically relevant $N$-soliton solutions of the BO equation can be obtained by a simple pole ansatz where the dynamics of the poles is governed by an $A_{N-1}$ CM system \cite{chen1979}. 
Thus, the BO equation on the real line is related to the rational CM system, while the BO equation with periodic boundary conditions is related to the trigonometric CM system.

The HWM equation is a nonlinear partial integro-differential equation that describes a spin density, represented by an $S^2$-valued function, propagating in one dimension. As two of us found recently in collaboration with Klabbers  \cite{berntsonklabbers2020}, the HWM equation is similar to the BO equation in that it has $N$-soliton solutions obtained via a spin-pole ansatz where the dynamics of the poles and of the spin degrees of freedom are determined by the $A_{N-1}$ sCM system and where, again, the real-line and periodic problems for the HWM equation correspond to the rational and trigonometric cases of this sCM system, respectively; see also \cite{ matsuno2022}. 
 It is known that the BO equation is related to a hydrodynamic description of the $A_{N-1}$ CM system \cite{abanov2009}, and that the HWM equation can be derived as a continuum limit of a spin chain which corresponds to a limiting case of the $A_{N-1}$ sCM system \cite{zhou2015}. This suggested to us that there should exist a more general soliton equation which includes the BO and HMW equations as limiting cases. As shown in the next two paragraphs, the sBO equation \eqref{eq:sBO} is such an equation. We found the sBO equation by generalizing the logic in \cite{abanov2009} to the sCM system, making use of a known B\"acklund transformation of $A$-type sCM systems \cite{gibbons1983}.  

The BO equation is given by 
\begin{equation} 
\label{eq:BO} 
u_t+2uu_x+Hu_{xx}=0 
\end{equation} 
for an $\R$-valued function $u=u(x,t)$. We make the ansatz $\mU(x,t)=u(x,t)\mP$ with $\mP$ a non-zero hermitian matrix satisfying $\mP^2=\mP$; then the sBO equation \eqref{eq:sBO} is satisfied if and only if $u(x,t)$ satisfies the BO equation \eqref{eq:BO}.  It is also interesting to note that by restricting $\mU$ to diagonal matrices, the sBO equation is reduced to $d$ decoupled BO equations: $\mU=\mathrm{diag}(u_1,\ldots,u_d)$ satisfies \eqref{eq:sBO} if and only if $u_\mu=u_\mu(x,t)$ solves the BO equation \eqref{eq:BO} for $\mu=1,\ldots,d$. 

The HWM equation describes the time evolution of an $\R^3$-valued function $\bm=(m^1,m^2,m^3)=\bm(x,t)$ satisfying the constraint $\bm^2\coloneqq (m^1)^2+(m^2)^2+(m^3)^2=1$ as follows, 
\begin{equation}
\label{eq:HWM}  
\bm_t = \bm\wedge H \bm_x 
\end{equation} 
where $\bm\wedge\bn \coloneqq (m^2n^3-m^3n^2,m^3n^1-m^1n^3,m^1n^2-m^2n^1)$ is the usual wedge product of three-vectors. By scaling $\mU(x,t)\to \lambda\mU(x,2\lambda t)$ and changing variables $2\lambda t\to t$ with $\lambda>0$ a scaling parameter, the sBO equation \eqref{eq:sBO} becomes 
\begin{equation} 
\label{eq:sBO1} 
\mU_t + \frac12 \{\mU,\mU_x\} + \frac1{2\lambda}H\mU_{xx}+\frac{\ii}{2}[\mU,H\mU_x] =0 . 
\end{equation}   
This reduces to a generalization of the HWM equation in the limit $\lambda\to\infty$ if we  impose the constraint $\mU^2=I$, where $I$ denotes the identity matrix.  
Indeed, $\mU^2=I$ implies $ \{\mU,\mU_x\}=0$, and by specializing $\mU$ to $2\times 2$ traceless hermitian matrices using the parametrization 
\begin{equation} 
\mU = \bm\cdot\bsigma = \begin{pmatrix} m^3 & m^1-\ii m^2\\ m^1+\ii m^2 & -m^3 \end{pmatrix} 
\end{equation} 
with $\bsigma=(\sigma_1,\sigma_2,\sigma_3)$ the Pauli matrices,  \eqref{eq:HWM} is obtained from \eqref{eq:sBO1} in the limit $\lambda\to\infty$, while $\mU^2=I$ is equivalent to $\bm^2=1$. We note that, while this reduction of the sBO equation to the HWM equation is mathematically simple, there is another similar but more complicated reduction explained in (ii) below which is more interesting from a physics point of view. 

\paragraph{(ii) Physics applications.} While hydrodynamics was initially developed to describe the propagation of fluids, recent developments in physics have established that hydrodynamic equations can provide a powerful tool to compute transport properties of strongly correlated electron systems like the cuprates or graphene; see e.g.\ \cite{hartnoll2007,andreev2011,svintsov2013}. Moreover, there exists a variety of topological such systems where the physical behavior is independent of model details and, in such a situation,  one can expect a successful description by integrable hydrodynamic equations (well-known arguments to explain this relation between universality and integrability in the context of soliton equations were given by Calogero \cite{calogero1991}); as an example, we mention the use of the BO equation to describe nonlinear waves at the boundary of fractional quantum Hall effect systems \cite{bettelheim2006,wiegmann2012}. 

Real electrons have spin and, for this reason, standard hydrodynamic equations describing the time evolution of a scalar density can only account for situations where the electron spin can be ignored. While this is the case for many conventional fractional quantum Hall effect systems, there also exist interesting such systems where the electron spin is important \cite{senthil1999,mccann2006,bernevig2006}. For such a system, one is interested in a hydrodynamic description by a soliton equation describing the time evolution of a fluid of particles carrying spin. The sBO equation \eqref{eq:sBO}, in the simplest non-trivial case when $\mU$ is a hermitian $2\times 2$ matrix, is a natural candidate for such an equation: as shown below, the sBO equation in this case can be written as a coupled system describing the time evolution of a charge- and a spin-density. Thus, we believe that it would be interesting to investigate if (a quantum version of) the sBO equation can be used to describe phenomena observed in quantum Hall effect systems where spin cannot be ignored, in generalization of results in \cite{bettelheim2006,wiegmann2012}. We recently proposed that the ncILW equation is relevant for parity invariant fractional Hall effect systems \cite{berntson2020}, and this suggests to us that its spin generalization presented in this paper (see \eqref{eq:sncILW}) will find applications in the context of the quantum spin Hall effect \cite{bernevig2006}.  

As another possible application of the sBO equation in physics, we mention the relation of the BO equation to conformal field theory of spin-less fermions \cite{abanov2005}. While spin-less fermions\footnote{To be more precise: spin-less chiral fermions in $1+1$ spacetime dimensions.} are among the simplest examples of a conformal field theory, there are conformal field theories that are natural spin generalizations of these models known as Wess-Zumino-Witten models \cite{difrancesco1997} which can take electron spin (and more complicated internal degrees of freedom) into account. We  expect that, in a similar way as the BO equation is related to spin-less fermions  \cite{abanov2005}, the sBO equation can be related to Wess-Zumino-Witten models. It would be interesting to substantiate this expectation. 

We now rewrite the sBO equation in the case where $\mU$ is a hermitian $2\times 2$ matrix as a coupled system describing the time evolution of a charge density $u$ and a spin density $\bm$. For that, we parametrize $\mU$ as follows, 
\begin{equation} 
\label{eq:mUansatz}
\mU = \frac{u}{2}\left(  I + \bm\cdot\bsigma\right)  = \frac{u}{2} \begin{pmatrix} 1 +m^3 & m^1-\ii m^2\\ m^1+\ii m^2 & 1-m^3 \end{pmatrix} 
\end{equation} 
where $u=u(x,t)$ and $\bm=(m^1,m^2,m^3)=\bm(x,t)$ are $\R$- and $\R^3$-valued functions, respectively. 
With this parametrization, we find after some computations that \eqref{eq:sBO} is equivalent to
\begin{equation} 
\label{eq:utbmt} 
\begin{split} 
u_t + (1+\bm^2)uu_x +\bm\cdot\bm_x u^2 +Hu_{xx} &= 0 ,
\\
\bm_t + u_x\bm(1-\bm^2)+u[\bm_x-\bm(\bm\cdot\bm_x)] +\frac1u[H(u\bm)_{xx} -\bm Hu_{xx}]   -\bm\wedge H(u\bm)_{x} 
 &=\bzero .
\end{split} 
\end{equation}
This system combines and generalizes the physics of the BO equation and of the HWM equation in a non-trivial way, preserving the exact solvability. Indeed, setting $\bm(x,t)=\bm_0$ (constant) such that $\bm_0^2=1$, the first equation in \eqref{eq:utbmt} becomes the BO equation \eqref{eq:BO}, while the second equation is trivially fulfilled. On the other hand, setting $u(x,t)=u_0$ (constant) and transforming $\bm(x,t)\to \bm(x-u_0t,u_0t)$, the first equation in  \eqref{eq:utbmt} is satisfied if we impose the condition $\bm^2=1$ and, with that, the second equations becomes the HWM equation \eqref{eq:HWM} in the limit $u_0\to \infty$. This suggests that \eqref{eq:utbmt} can be well approximated by the BO equation if the variation of the spin density in space and time can be ignored, while the HWM equation is a good approximation to \eqref{eq:utbmt} (on an appropriate time scale) if the charge density is large and only deviates slightly from a constant background $u_0$.

\paragraph{(iii) Multi-soliton solutions.} Following Ref.~\cite{gibbons1984}, we use the Dirac bra-ket notation \cite{dirac1939} and denote by $|e\rangle$ and $\langle f|$ vectors in some $d$-dimensional complex vector space $\cV$ and its dual $\cV^*$,  respectively; in particular, $|e\rangle\langle f|$ represents a $d\times d$ matrix with complex entries, and $|e\rangle\langle f|^\dag = |f\rangle\langle e|$ is the hermitian conjugate of this matrix. Moreover, $*$ is complex conjugation. (Readers not familiar with this notation can identify $|e\rangle\in \cV$ with $(e_{\mu})_{\mu=1}^d\in\C^d$, $\langle f|\in \cV^*$ with  $(f^*_{\mu})_{\mu=1}^d\in\C^d$, $\langle f|e \rangle$ with the scalar product $\sum_{\mu=1}^d f_\mu^*e^{\phantom*}_\mu$, and $|e\rangle\langle f|$ with the matrix $(e^{\phantom*}_\mu f_\nu^*)_{\mu,\nu=1}^d$.) 

We show in Section~\ref{sec:sBO} that, for arbitrary integer $N\geq 1$ and $x\in\R$,  
\begin{equation}
\label{eq:mUintro}  
\mU(x,t) = \ii \sum_{j=1}^N |e_j(t)\rangle \langle f_j(t)| \frac{1}{x-a_j(t)} - \ii \sum_{j=1}^N |f_j(t)\rangle \langle e_j(t)| \frac{1}{x-a_j(t)^*}
\end{equation} 
is a solution of the sBO equation \eqref{eq:sBO} provided that the variables $a_j=a_j(t)\in\C$, $|e_j\rangle=|e_j(t)\rangle\in\cV$ and $\langle f_j|=\langle f_j(t)|\in\cV^*$ satisfy the following time evolution equations,\footnote{$\sum_{k\neq j}^N$ is short for $\sum_{k=1,k\neq j}^N$.} 
\begin{equation} 
\label{eq:sCMintro} 
\begin{split} 
\frac{\dd^2}{\dd t^2} a_j = & \, 8\sum_{k\neq j}^N\frac{\langle f_j|e_k\rangle\langle f_k|e_j\rangle }{(a_j-a_k)^3},\\
\frac{\dd}{\dd t}|e_j\rangle = &\, 2\ii\sum_{k\neq j}^N \frac{|e_k\rangle\langle f_k|e_j\rangle}{(a_j-a_k)^2},\\ 
\frac{\dd}{\dd t}\langle f_j|  = &\,  -2\ii\sum_{k\neq j}^N\frac{\langle f_j|e_k\rangle \langle f_k|}{(a_j-a_k)^2} 
\end{split} 
\end{equation} 
for $j=1,\ldots,N$, with initial conditions that satisfy certain constraints; see Theorem~\ref{thm:sBO} for a precise formulation. 
Moreover, we show that initial conditions satisfying the pertinent constraints can be constructed by solving a linear algebra problem involving a $N d\times N d$ hermitian matrix, and that the solution \eqref{eq:mUintro} depends on $N d$ complex parameters; see Section~\ref{sec:NsolitonssBO} for details. 

It is important to note that, up to a rescaling of time, the time evolution equations \eqref{eq:sCMintro} and one of the above-mentioned constraints on the initial conditions (given in \eqref{eq:fjej}) defines the sCM model discovered by Gibbons and Hermsen \cite{gibbons1984}, and another constraint (given in \eqref{eq:BTt=0}) is a known B\"acklund transformation of the sCM model \cite{gibbons1983}.
Moreover, for $d=1$, the solution above reduces to the $N$-soliton solution of the BO equation found in \cite{chen1979} that relates the BO equation to the rational $A_{N-1}$ CM model. Thus, in the same sense as the sCM model is a natural generalization of the simplest non-trivial CM model, the sBO equation is a natural generalization of the BO equation. 

As shown by Gibbons and Hermsen in their original paper \cite[Section 6]{gibbons1984}, the sCM model describes the time evolution of poles and spins of rational solutions of the boomeron equation, which is a soliton equation introduced and studied in \cite{calogero1976coupled}. One might wonder if there is a relation between this fact and our work. We cannot exclude this possibility: it is conceivable to us that it is possible to derive the boomeron equation as a local limit $\delta\to\infty $ of the sncILW equation \eqref{eq:sncILW} or the sILW equation \eqref{eq:sILW} introduced below and, in this way, there could be an indirect relation. However, even if this is the case, we expect that it would be challenging to make this relation precise. Indeed, while the boomeron equation is similar to \eqref{eq:utbmt} in that it is a system describing the time evolution of a vector coupled to a scalar, it is different in other important ways; in particular, \eqref{eq:utbmt} is rotation invariant (and this is also the case for the certain local limits of these equations that we derive, see Section~\ref{sec:locallimit} for details), while the boomeron equation is not.

\paragraph{(iv) Spin generalization of ncILW equation and soliton-CM correspondence.} For $\delta>0$, we define the following integral operators acting on functions $f$ of $x\in\R$, 
\begin{equation} 
\label{eq:TT}
\begin{split} 
(Tf)(x) &= \frac1{2\delta}\pvint_{\R}\coth\left(\frac{\pi}{2\delta}(x'-x)\right)f(x')\,\dd{x'}, \\
(\tT f)(x) &= \frac1{2\delta}\int_{\R}\tanh\left(\frac{\pi}{2\delta}(x'-x)\right)f(x')\,\dd{x'}.
\end{split} 
\end{equation} 
The ncILW equation was introduced in \cite{berntson2020} as a non-chiral version of the ILW equation and is given for two scalar-valued functions $u(x,t)$ and $v(x,t)$ by
\begin{equation} 
\label{eq:ncILW} 
\begin{split} 
&u_t + 2 u u_x + Tu_{xx}+\tilde{T}v_{xx}=0,\\
&v_t - 2 v v_x - Tv_{xx}-\tilde{T}u_{xx}=0.
\end{split} 
\end{equation} 
Here, we introduce the following spin generalization of the ncILW equation given for two square matrix-valued functions $\mU=\mU(x,t)$ and $\mV=\mV(x,t)$ by
\begin{equation} 
\label{eq:sncILW} 
\boxed{
\begin{aligned} 
\mU_t &+ \{\mU,\mU_x\} + T\mU_{xx}+\tT \mV_{xx} +\ii [\mU,T\mU_x]+\ii [\mU,\tT \mV_x] =0,\\
\mV_t &- \{\mV,\mV_x\} - T\mV_{xx}-\tT \mU_{xx} +\ii [\mV,T\mV_x]+\ii[\mV,\tT \mU_x] =0
\end{aligned} 
}
\end{equation} 
(again, the matrix size $d\times d$ of $\mU$ and $\mV$ is arbitrary; $d=1$ corresponds to the ncILW equation \eqref{eq:ncILW}).
We refer to \eqref{eq:sncILW} as the {\em spin ncILW (sncILW) equation}. 
Our main result on the sncILW equation is a construction of $N$-soliton solutions obtained via a spin-pole ansatz where the dynamics of the spins and poles are determined by the $A_{N-1}$ sCM model in the hyperbolic case (III), in natural generalization of a known result about the ncILW equation \cite{berntson2020}; see Theorem~\ref{thm:sncILW}. 

Note that $\lim_{\delta\to+\infty}\tT =0$ \cite{berntson2020} and thus, clearly, the two equations in \eqref{eq:sncILW} decouple in the limit $\delta\to\infty$. 
Moreover, since the limit $\delta\to\infty$ of $T$ coincides with the Hilbert transform $H$ given in \eqref{eq:H} \cite{kodama1981} (see also \cite{berntson2020}), the first of these decoupled equations is identical to the sBO equation \eqref{eq:sBO}, and the second equation is obtained from the sBO equation  \eqref{eq:sBO} by the parity transformation $\mU(x,t)\to \mV(x,t)=\mU(-x,t)$. Since the sBO equation and the one obtained from it by this parity transformation are different, the sBO equation   \eqref{eq:sBO} is chiral, and we therefore regard \eqref{eq:sncILW} as a non-chiral generalization of the sBO equation where the two chiral degrees of freedom, $\mU$ and $\mV$, are coupled by the operator $\tT$.

We also introduce periodic versions of the sBO and sncILW equations: the former is given by \eqref{eq:sBO} for an $L$-periodic, square matrix-valued function $\mU$, $\mU(x+L,t)=\mU(x,t)$ with $L>0$ a fixed parameter, and the Hilbert transform
\begin{equation}
\label{eq:Hp}  
(Hf)(x)\coloneqq \frac1{L}\pvint_{-L/2}^{L/2} \cot\left(\frac{\pi}{L}(x'-x) \right)f(x')\,\dd{x'}; 
\end{equation} 
the latter is given by \eqref{eq:sncILW} for $L$-periodic, square matrix-valued functions $\mU$ and $\mV$ and the integral operators 
\begin{equation} 
\label{eq:TTe}
\begin{split} 
(Tf)(x) &= \frac1{\pi}\pvint_{L/2}^{L/2}\zeta_1(x'-x)f(x')\,\dd{x'}, \\
(\tT f)(x) &= \frac1{\pi}\int_{-L/2}^{L/2}\zeta_1(x'-x+\ii\delta)f(x')\,\dd{x'},
\end{split} 
\end{equation} 
where 
\begin{equation}
\label{eq:zeta1} 
\zeta_{1}(z)\coloneqq  \zeta(z)-\frac{\eta_{1}}{\omega_{1}}z
\end{equation} 
with $\zeta(z)$ the Weierstrass $\zeta$-function with half-periods $(\omega_1,\omega_2)=(L/2,\ii\delta)$ and $\eta_1=\zeta(\omega_1)$ \cite{DLMF}. 
Similarly as in the real-line case, the periodic sBO equation is chiral, and the periodic sncILW equation reduces to two decoupled periodic sBO equations of opposite chirality in the limit $\delta\to\infty$ (details can be found in \cite{berntson2020}).  

To explain the relation between the soliton equations discussed above and CM systems, we recall that each CM system comes in four versions which can be distinguished by the special function
\begin{equation} 
\label{eq:alpha} 
\alpha(z)\coloneqq \begin{cases} 1/z & \text{(I: rational case)}\\
(\pi/L)\cot(\pi z/L) & \text{(II: trigonometric case)}\\
(\pi/2\delta)\coth(\pi z/2\delta) & \text{(III: hyperbolic case)}
\end{cases} \quad (z\in\C) 
\end{equation} 
with $L>0$ and $\delta>0$ fixed parameters;  the fourth version corresponds to 
\begin{equation} 
\alpha(z)\coloneqq \zeta_1(z) \quad  \text{(IV: elliptic case)}, 
\end{equation} 
but we do not include it in \eqref{eq:alpha} since it is more complicated and, for this reason, our results in this paper are restricted to cases I--III.  
This function $\alpha(z)$ is important since, for example, it determines the corresponding CM interaction potential as $V(z)=-\alpha'(z)$ \cite{olshanetsky1981}. 
Moreover, using these special functions, one can define integral operators
\begin{equation} 
\label{eq:TTalpha}
\begin{split} 
(Tf)(x) &\coloneqq \frac1\pi \pvint \alpha(x'-x)f(x')\,\dd{x'} \quad \text{(cases I--IV)}, \\
(\tT f)(x)&\coloneqq \frac1\pi \pvint \alpha(x'-x+\ii\delta)f(x')\,\dd{x'} \quad \text{(cases III and IV)}, 
\end{split} 
\end{equation} 
where the integrations are over $\R$ in cases I and III and over $[-L/2,L/2]$ in cases II and IV. 
We note that, in cases I and II, $T$ is identical to the Hilbert transform $H$ in \eqref{eq:H} and \eqref{eq:Hp}, respectively, suggesting that the real-line and periodic versions of the sBO equation are related to the $A$-type sCM system in the rational and trigonometric cases, respectively. Similarly, in cases III and IV, the operators $T$ and $\tT$ in \eqref{eq:TTalpha} are identical to the operators in \eqref{eq:TT}  and \eqref{eq:TTe}, suggesting that the real-line and periodic versions of the sncILW equation are related to the $A$-type sCM system in the hyperbolic and elliptic cases, respectively. The $N$-soliton solutions of these equations obtained in this paper confirm these expectations in cases I--III, and we conjecture that this result can be generalized to case IV. Thus, the equations proposed in the present paper extend the relation between soliton equations and CM systems proposed in \cite{berntson2020} to the spin setting. 


As shown by two of us in collaboration with Klabbers \cite{berntsonklabbers2021}, there exists a non-chiral variant of an intermediate generalization of the Heisenberg ferromagnetic equation which generalizes the HWM equation and which has soliton solutions given by a spin-pole ansatz governed by the hyperbolic sCM model, in generalization of a result in \cite{berntsonklabbers2020} for the HWM equation mentioned above; this {\em non-chiral intermediate Heisenberg ferromagnetic equation} is given by 
\begin{equation} 
\begin{split} 
\bm_t & =\, + \bm\wedge T\bm_{x} - \bm\wedge \tT\bn_{x},\\
\bn_t & =\, - \bn\wedge T\bn_{x} + \bn\wedge \tT\bm_{x},\\
\end{split} 
\end{equation} 
for two $\R^3$-valued functions $\bm=\bm(x,t)$ and $\bn=\bn(x,t)$ satisfying $\bm^2=\bn^2=1$  \cite{berntsonklabbers2021}, and we checked that it is obtained from the sncILW equation \eqref{eq:sncILW} in a similar way as the HWM equation is obtained from the sBO equation \eqref{eq:sBO} (as explained in the paragraph containing \eqref{eq:sBO1}), after changing the sign of $\mV$ (the latter is merely a convention).

Our results in the present paper, together with results in the literature on the standard (chiral) ILW equation \cite{kodama1981} and on the non-chiral intermediate Heisenberg ferromagnet equation \cite{berntsonklabbers2021}, suggest that the following is also an integrable generalization of the sBO equation, 
\begin{equation} 
\label{eq:sILW} 
\boxed{ 
\mU_t +  \{\mU,\mU_x\} + \frac1\delta\mU_x  + T\mU_{xx}  + \ii [\mU,T\mU_x] =0
}
\end{equation} 
(we added the term $\mU_x/\delta$ for convenience; note that this term can be removed by a Galilean transformation $x\to x-t/\delta$). We refer to \eqref{eq:sILW}  as the {\em spin ILW (sILW) equation}. Note that, in the limit $\delta\to\infty$, it reduces to the sBO equation \eqref{eq:sBO}. We show that, in generalization of the well-known fact that the standard ILW equation reduces to the Korteweg-de Vries  (KdV) equation in a limit $\delta\downarrow 0$ (see e.g. \cite{scoufis2005}), 
 \eqref{eq:sILW}  reduces to the following matrix generalization of the KdV equation in a certain $\delta\downarrow 0$ limit, 
\begin{equation} 
\label{eq:mKdV} 
\mU_t +  \{\mU,\mU_x\} + \mU_{xxx}  = 0  
\end{equation} 
(see Section~\ref{sec:locallimit}). 
This so-called {\em matrix KdV equation} was introduced by Lax \cite{lax1968}, and its multisoliton solutions were constructed in \cite{goncharenko2001}. 
Moreover, the sILW equation allows for another limit $\delta\downarrow 0$ leading to the following generalization of the Heisenberg ferromagnet (HF) equation 
\begin{equation} 
\label{eq:sHF} 
\mU_t + \ii [\mU,\mU_{xx}]=0
\end{equation} 
with the constraint $\mU^2=I$, where $I$ denotes the identity matrix (see Section~\ref{sec:locallimit} for details); note that \eqref{eq:sHF} reduces to the standard HF equation if $\mU$ is restricted to lie in the class of $2\times 2$ hermitian traceless matrices. 
Thus, the sILW equation is an interpolation between the matrix KdV equation, the HF equation, and the sBO equation. 
 
 We call the equations introduced in the present paper the {\em spin} (rather than the {\em matrix}) BO and (nc)ILW equations since, in contrast to the matrix KdV equation \eqref{eq:mKdV}, they  incorporate the nonlinear physics of charge- and spin-densities in a single equation; as discussed above, this property is lost in the local limit $\delta\to\infty$ (since, depending on the scaling, either the Heisenberg-term $\ii [\mU,\mU_{xx}]$ or the KdV-terms $\{\mU,\mU\}_x+\mU_{xxx}$ disappear in that limit).

\paragraph{(v) Summary and open problems.} 
In this paper, we introduce a family of equations containing well-known soliton equations as special and limiting cases. This family of equations consists of the sBO equation \eqref{eq:sBO}, the sncILW equation \eqref{eq:sncILW}, and the sILW equation \eqref{eq:sILW} for different matrix sizes $d\in\Z_{\geq 2}$; for $d=1$, these equations reduce to the known  BO, ncILW, and ILW equations, respectively.  
We show that the sBO and sncILW equations are exactly solvable and, by that, establish a relation to sCM models.  While we do not give any result about the sILW equation for $0<\delta<\infty$, we believe that the arguments we present about this equation strongly suggest that it is integrable.
Obviously, it would be interesting to generalize other known results about the special case $d=1$ (e.g., Hirota bilinear forms, Lax pairs, B\"acklund transformations, and inverse scattering transforms) to $d>1$.

For the ncILW equation, we were able to generalize the multi-soliton solutions in the real-line case to the periodic case in \cite{berntson2020}. 
However, for the sncILW equation \eqref{eq:sncILW}, we do not have this result. We found that the construction of soliton solutions of the sncILW equation in the periodic case is more challenging for $d>1$ than for $d=1$ for the following reasons: first, the generalization of Proposition~\ref{prop:BT} to the elliptic case (IV) is more difficult (for $d=1$, this generalization is known), and second, the commutator terms in \eqref{eq:sncILW} lead to severe complications. We thus believe that the construction of soliton solutions of the periodic sncILW equation is an interesting problem requiring new ideas.

We were inspired to search for the soliton equations presented in this paper by heuristic arguments suggested to us by the known relation between the quantum version of the BO equation and the Calogero-Sutherland model\footnote{This is the quantum version of the trigonometric $A$-type CM model.} \cite{abanov2005}, together with a generalization of this relation to the corresponding elliptic Calogero-Sutherland model that lead us to the ncILW equation \cite{berntson2020}. 
It would be interesting to promote these heuristic arguments to precise results by constructing a second quantization of the quantum versions of the trigonometric sCM models, in generalization of results in \cite{carey1999}. We believe that this can open a way towards finding a precise relation between the sBO equation and Wess-Zumino-Witten models.

As already discussed in (ii) above, there are several systems in the real world motivating the development of hydrodynamic descriptions of identical particles with spin; however, hydrodynamics for coupled charge- and spin-densities is not a fully developed subject. As one specific example,  we mention recent work deriving such hydrodynamic equations using the known hydrodynamic description of the standard CM model as a guide \cite{kulkarni2009}; see also Xing's PhD thesis \cite{xing2015} aiming in this direction.
We hope that the present paper opens up a way to push these results further towards a hydrodynamic description of the sCM model which can serve as a prototype of spinful hydrodynamics.

\paragraph{Plan of paper.} We state and prove the B\"acklund transformations for the sCM model in the rational, trigonometric and hyperbolic cases in Section~\ref{sec:sCM}. In Section~\ref{sec:sBO} and \ref{sec:sncILW}, we state and prove our multi-soliton solutions of the sBO and sncILW equation, respectively; we note that while the results in Section~\ref{sec:sBO} on the sBO equation are for the real-line and periodic cases, the results in Section~\ref{sec:sncILW} on the sncILW equation only apply to the real-line case.  Section~\ref{sec:results} contains the Hamiltonian formulations of the sBO and sncILW equations (Section~\ref{sec:Hamiltonian}), details on the local limit $\delta\to\infty$ of the sILW equation (Section~\ref{sec:locallimit}), and the spin generalization of the bidirectional BO equation that led us to the results in the present paper (Section~\ref{sec:2sBO}). We also include three appendices with functional identities that we need (Appendix~\ref{app:alphaV}), and the non-hermitian solutions of the sBO equation (Appendix~\ref{app:sBOsolutions}) and the sncILW equation (Appendix~\ref{app:sncILWsolutions}), respectively.

\section{Spin Calogero-Moser systems}
\label{sec:sCM}  
In this section, we collect results about the $A$-type sCM systems due to Gibbons and Hermsen \cite{gibbons1984} that we need. In particular, we define the $A$-type sCM systems in Section~\ref{sec:sCM_def}, and, following  \cite{gibbons1983}, we state and prove a B\"acklund transformation for these systems in Sections~\ref{sec:sCM_BT} and \ref{sec:BT_proof}. For simplicity, we restrict our discussion to cases I--III; we expect a similar result for the elliptic case IV, but this case is more complicated and thus left to future work. 

We believe that the results in this section are of interest in their own right: to our knowledge, the proof of the relevant B\"acklund transformation (see \eqref{eq:BT}) in the literature has previously been restricted to the rational case (I) and $M=N$ \cite{gibbons1983}. We mention in passing that, while the original paper introduced and solved the sCM model only in the rational case \cite{gibbons1984}, the integrability of the sCM model in all cases I--IV was proved in \cite{hikami1993}, and explicit solutions of the sCM model in the elliptic case (IV) can be found in  \cite{krichever1995}. 

\subsection{Definition}
\label{sec:sCM_def}
Let $N\in\Z_{\geq 1}$ be arbitrary and let
\begin{equation} 
\label{eq:V} 
V(z)\coloneqq \begin{cases} 1/z^2 & \text{(I: rational case)}\\
(\pi/L)^2/\sin^2(\pi z/L) & \text{(II: trigonometric case)}\\
(\pi/2\delta)^2/\sinh^2(\pi z/2\delta) & \text{(III: hyperbolic case).}\\
\end{cases} \quad (z\in\C) 
\end{equation} 
For each case I--III, the corresponding $A_{N-1}$ sCM system is a dynamical system of $N$ particles moving in the complex plane and with internal degrees of freedom described by a $d$-dimensional vector space $\cV$ and its dual $\cV^*$, with $d\in\Z_{\geq 1}$ arbitrary. 
Denoting the position of the $j$th particle at time $t\in\R$ by $a_j=a_j(t)\in\C$ and its internal degrees of freedom by vectors $|e_j\rangle=|e_j(t)\rangle\in \cV$ and $\langle f_j|=\langle f_j(t)|\in \cV^*$, this system can be defined by the time evolution equations
\begin{equation} 
\label{eq:sCM1a} 
\ddot a_j = -4\sum_{k\neq j}^N \langle f_j|e_k\rangle\langle f_k|e_j\rangle V'(a_j-a_k) \quad (j=1,\ldots,N) 
\end{equation} 
and
\begin{equation} 
\begin{split} 
\label{eq:sCM1b} 
|\dot e_j\rangle &= 2\ii\sum_{k\neq j}^N |e_k\rangle\langle f_k|e_j\rangle V(a_j-a_k),\\
\langle\dot f_j| & = -2\ii\sum_{k\neq j}^N\langle f_j|e_k\rangle \langle f_k| V(a_j-a_k)
\end{split} \quad (j=1,\ldots,N) 
\end{equation} 
(the dot indicates differentiation with respect to $t$ and the prime indicates differentiation with respect to the argument of the respective function), together with the following constraints, 
\begin{equation} 
\label{eq:sCM1c} 
\langle f_j|e_j\rangle = 1 \quad  (j=1,\ldots,N) 
\end{equation}  
(our notation is explained in the paragraph above \eqref{eq:mUintro}).
Observe that the constraints \eqref{eq:sCM1c} are preserved under the equations of motion \eqref{eq:sCM1b}.

\subsection{B\"acklund transformations}
\label{sec:sCM_BT}
We consider the sCM system \eqref{eq:sCM1a}--\eqref{eq:sCM1c}, together with another such system involving $M\in\Z_{\geq 0}$ particles located at the positions $b_j=b_j(t)\in\C$, $j = 1, \dots, M$, together with spin degrees of freedom $|g_j\rangle=|g_j(t)\rangle\in \cV$ and $\langle h_j|=\langle h_j(t)|\in \cV^*$ (note that the vector spaces $\cV$ and $\cV^*$ are the same as for the first system; while $M=N$ is an important special case, also the cases $M\neq N$ and, in particular, $M=0$ are interesting). 
More specifically, the second system is given by the time evolution equations
\begin{equation} 
\label{eq:sCM2a} 
\ddot b_j = -4\sum_{k\neq j}^M \langle h_j|g_k\rangle\langle h_k|g_j\rangle V'(b_j-b_k) \quad (j=1,\ldots,M) 
\end{equation} 
and 
\begin{equation} 
\begin{split} 
\label{eq:sCM2b} 
|\dot g_j\rangle &= 2\ii\sum_{k\neq j}^M |g_k\rangle\langle h_k|g_j\rangle V(b_j-b_k),\\
\langle\dot h_j| & = -2\ii\sum_{k\neq j}^M\langle h_j|g_k\rangle \langle h_k| V(b_j-b_k), 
\end{split} \quad (j=1,\ldots,M) 
\end{equation} 
and the constraints 
\begin{equation} 
\label{eq:sCM2c} 
\langle h_j|g_j\rangle = 1 \quad (j=1,\ldots,M) .
\end{equation}  

As shown in \cite{gibbons1983} in a special case, two such sCM systems \eqref{eq:sCM1a}--\eqref{eq:sCM1c} and  \eqref{eq:sCM2a}--\eqref{eq:sCM2c} are connected by a B\"acklund transformation as follows, 
\begin{subequations} 
\label{eq:BT} 
\begin{align} 
\label{eq:BTa} 
\dot a_j \langle f_j| &= 2\ii\sum_{k\neq j}^N  \langle f_j|e_k\rangle \langle f_k|\alpha(a_j-a_k) -2\ii\sum_{k=1}^M  \langle f_j|g_k\rangle\langle h_k|\alpha(a_j-b_k)\quad &(j=1,\ldots,N), \\
\label{eq:BTb} 
\dot b_j |g_j\rangle &=-2\ii\sum_{k\neq j}^M |g_k\rangle\langle h_k|g_j\rangle\alpha(b_j-b_k)+2\ii\sum_{k=1}^N|e_k\rangle\langle f_k|g_j\rangle\alpha(b_j-a_k) \quad &(j=1,\ldots,M),    
\end{align} 
\end{subequations} 
with the function $\alpha(z)$ given by \eqref{eq:alpha}. The precise statement is given below. 

\begin{proposition}[B\"acklund transformation for sCM system]
\label{prop:BT} 
In each case I--III, the first order equations  \eqref{eq:sCM1b}, \eqref{eq:sCM2b} and \eqref{eq:BT}, together with the constraints \eqref{eq:sCM1c} and  \eqref{eq:sCM2c}, imply the second order equations \eqref{eq:sCM1a} and \eqref{eq:sCM2a}.  
\end{proposition} 
A self-contained proof of Proposition \ref{prop:BT} can be found in Section~\ref{sec:BT_proof}.

\begin{remark} 
As already mentioned, the special case $M=N$ in the rational case (I) was stated and proved in  \cite{gibbons1983}, More specifically, in this special case, the equations \eqref{eq:BTa} and \eqref{eq:BTb} reduce to the second and first equations in \cite[Eq.~(17)]{gibbons1983}, respectively, using the transformation $t\to -2t$ and the identifications 
\begin{equation} 
 (a_j,\dot a_j,|e_j\rangle,\langle f_j|,b_j,\dot b_j,|g_j\rangle,\langle h_j|) \to x^+_j,p^+_j,|e^+_j\rangle,\langle f^+_j|,x_j,p_j,|e_j\rangle,\langle f_j|) \quad (j=1\ldots,N).
\end{equation}  
\end{remark} 

It is interesting to note that Proposition~\ref{prop:BT} has the following consistent reduction when $N=M$,
\begin{equation} 
\label{eq:reduction} 
b_j=a_j^*,\quad |g_j\rangle = \langle f_j|^\dag = |f_j\rangle,\quad \langle h_j| = |e_j\rangle^\dag = \langle e_j|\quad (j=1,\ldots,N)
\end{equation} 
where $*$ and $\dag$ indicate complex and hermitian conjugation, respectively. 
Indeed, by imposing these conditions, \eqref{eq:sCM2a}--\eqref{eq:sCM2c}  and \eqref{eq:BTb} become the hermitian conjugate of \eqref{eq:sCM1a}--\eqref{eq:sCM1c} and \eqref{eq:BTa}, respectively, and Proposition~\ref{prop:BT} simplifies as follows. 

\begin{corollary}
\label{cor:BT} 
In each case I--III, the first order equations \eqref{eq:sCM1b} and 
\begin{equation}
\label{eq:BThermitian} 
\dot a_j \langle f_j| = 2\ii\sum_{k\neq j}^N  \langle f_j|e_k\rangle \langle f_k|\alpha(a_j-a_k) -2\ii\sum_{k=1}^N  \langle f_j|f_k\rangle\langle e_k|\alpha(a_j-a^*_k)\quad (j=1,\ldots,N), 
\end{equation} 
together with the constraints \eqref{eq:sCM1c}, imply the  second order equations \eqref{eq:sCM1a}. 
\end{corollary} 

\subsection{Proof of Proposition~\ref{prop:BT}} 
\label{sec:BT_proof} 
We note that, in the system of equations \eqref{eq:sCM1a}--\eqref{eq:BT}, the sets of variables $\{a_j,|e_j\rangle,\langle f_j|\}_{j=1}^N$ and $\{b_j,\langle h_j|, |g_j\rangle\}_{j=1}^M$ can be swapped by hermitian conjugation and renaming $a_j^*\to a_j$, $b_j^*\to b_j$. Due to this symmetry, it suffices to verify the claim for the first set of variables, i.e., it is enough to show that \eqref{eq:sCM1a} follows from \eqref{eq:sCM1b}, \eqref{eq:sCM2b}, \eqref{eq:BT}, subject to \eqref{eq:sCM1c} and \eqref{eq:sCM2c}.
 
We introduce the shorthand notation 
\begin{equation} 
\label{eq:shorthandBT} 
(a_j,|e_j\rangle,\langle f_j|,r_j) = \begin{cases} (a_j,|e_j\rangle,\langle f_j|,+1) & (j=1,\ldots,N) \\ (b_{j-N},|g_{j-N}\rangle,\langle h_{j-N}|,-1) & (j=N+1,\ldots,\cN) , \end{cases} \quad 
\cN\coloneqq N+M, 
\end{equation}
and
\begin{equation} 
\label{eq:mBj} 
 \mP_j\coloneqq |e_j\rangle\langle f_j| , \quad \mB_j\coloneqq \ii\sum_{k\neq j}^{\cN} r_k \mP_k  \alpha(a_j-a_k),\quad (j=1,\ldots,\cN) 
\end{equation}  
to write \eqref{eq:BT} as  
\begin{equation}
\label{eq:BTshort} 
\begin{split}  
\dot a_j\langle f_j| = &\;  2\langle f_j|\mB_j\quad (j=1,\ldots,N),      \\
\dot a_j|e_j\rangle = &\; 2\mB_j|e_j\rangle \quad (j=N+1,\ldots,\cN).
\end{split}
\end{equation} 

Moreover, this notation allows us to write the two sets of equations \eqref{eq:sCM1a}--\eqref{eq:sCM1c} and \eqref{eq:sCM2a}--\eqref{eq:sCM2c} as one: 
\begin{equation} 
\label{eq:sCM3a} 
\ddot a_j = -2\sum_{k\neq j}^\cN (1+r_jr_k)\langle f_j|e_k\rangle\langle f_k|e_j\rangle V'(a_j-a_k) \quad (j=1,\ldots,\cN) 
\end{equation} 
and 
\begin{equation} 
\begin{split} 
\label{eq:sCM3b} 
|\dot e_j\rangle &= \ii\sum_{k\neq j}^\cN (1+r_jr_k) |e_k\rangle\langle f_k|e_j\rangle V(a_j-a_k),\\
\langle\dot f_j| & = -\ii\sum_{k\neq j}^\cN(1+r_jr_k)\langle f_j|e_k\rangle \langle f_k| V(a_j-a_k), 
\end{split} \quad (j=1,\ldots,\cN) , 
\end{equation} 
together with
\begin{equation} 
\label{eq:sCM3c} 
\langle f_j|e_j\rangle = 1 \quad (j=1,\ldots,\cN).  
\end{equation}  

By differentiating the first set of equations in \eqref{eq:BTshort} with respect to time, we obtain 
\begin{equation} 
\label{eq:ddotajfj1}
\ddot a_j\langle f_j| =   \langle \dot f_j|(2\mB_j-\dot a_j) + 2\langle f_j|\dot \mB_j
\end{equation} 
where, here and below in this section, $j=1,\ldots,N$. 
We compute, using \eqref{eq:sCM3b} and $|e_k\rangle\langle f_k|=\mP_k$ (note that $r_j=+1$), 
\begin{equation} 
\label{eq:dotfj} 
\begin{split} 
\langle \dot f_j|(2\mB_j-\dot a_j) = & -\ii\sum_{k\neq j}^\cN(1+r_k) \langle f_j|e_k\rangle\langle f_k |(2\mB_j-\dot a_j)V(a_j-a_k) \\
=  & -\ii\sum_{k\neq j}^\cN(1+r_k)\langle f_j|( 2\mP_k \mB_j-2\mB_j\mP_k)V(a_j-a_k)  \\
=  & -2\ii\sum_{k\neq j}^\cN(1+r_k)\langle f_j|[\mP_k, \mB_j]V(a_j-a_k), 
\end{split} 
\end{equation} 
inserting \eqref{eq:BTshort} in the second step. Moreover, the definition \eqref{eq:mBj} of $\mB_j$ and the relation $\alpha'(z)=-V(z)$  imply
\begin{equation} 
2\langle f_j|\dot \mB_j = 2\ii\sum_{k\neq j}^\cN  r_k \langle f_j|\dot\mP_k \alpha(a_j-a_k) - 2\ii\sum_{k\neq j}^\cN  r_k \langle f_j|\mP_k(\dot a_j-\dot a_k)V(a_j-a_k), 
\end{equation} 
and by using \eqref{eq:BTshort} we compute
\begin{equation} 
\label{eq:fjPkdotajdotak}
\begin{split} 
\langle f_j|\mP_k(\dot a_j-\dot a_k)  = &\; \langle f_j|e_k\rangle\langle f_k| (\dot a_j-\dot a_k) \\
= & \; \dot a_j\langle f_j|e_k\rangle\langle f_k|  - \frac12(1+r_k)\langle f_j|e_k\rangle\dot a_k \langle f_k| - \frac12(1-r_k)\langle f_j|e_k\rangle \dot a_k\langle f_k|\\
= &\;  2 \langle f_j|\mB_j\mP_k - (1+r_k)\langle f_j|\mP_k\mB_k -(1-r_k)\langle f_j|\mB_k\mP_k\\
= & \; \langle f_j|\big(\{\mP_k,\mB_j-\mB_k\} + r_k[\mP_k,\mB_j-\mB_k] -  (1+r_k)[\mP_k,\mB_j]\big) . 
\end{split} 
\end{equation} 
Inserting the results in \eqref{eq:dotfj}--\eqref{eq:fjPkdotajdotak} into \eqref{eq:ddotajfj1} and using  $r_k^2=1$,  we find that $ \langle \dot f_j|(2\mB_j-\dot a_j)$ is canceled by the part of  $2\langle f_j|\dot \mB_j$ involving the term $-(1+r_k)\langle f_j[\mP_k,\mB_j]$ in \eqref{eq:fjPkdotajdotak}, and we obtain  
\begin{equation} 
\label{eq:ddotajfj2}
\begin{split} 
\ddot a_j\langle f_j| = &\;   2\ii\sum_{k\neq j}^\cN  r_k \langle f_j|\dot\mP_k \alpha(a_j-a_k) \\
& -2\ii \sum_{k\neq j}^\cN \langle f_j|\big( r_k \{\mP_k,\mB_j-\mB_k\} + [\mP_k,\mB_j-\mB_k] \big)V(a_j-a_k) . 
\end{split} 
\end{equation} 
To proceed, we use $\mP_k=|e_k\rangle\langle f_k|$ and \eqref{eq:sCM3b} to compute 
\begin{equation} 
\label{eq:dotPk} 
\begin{split} 
\dot \mP_k = &\; |\dot e_k\rangle\langle f_k| + |e_k\rangle\langle \dot f_k| \\
= &\; \ii\sum_{l\neq k}^{\cN}(1+r_kr_l) ( |e_l\rangle\langle f_l | e_k\rangle\langle f_k| -  |e_k\rangle\langle f_k|e_l\rangle\langle f_l|)V(a_k-a_l)  \\
=  & -\ii\sum_{l\neq k}^{\cN}(1+r_kr_l) [\mP_k,\mP_l]V(a_k-a_l).
\end{split} 
\end{equation} 
For future use, it is convenient to rewrite \eqref{eq:dotPk} as 
\begin{equation}
\label{eq:dotPk2}
\dot \mP_k=- \ii(1+r_k)[\mP_k,\mP_j]V(a_j-a_k)-\ii\sum_{l\neq j,k}^{\cN}(1+r_kr_l) [\mP_k,\mP_l]V(a_k-a_l),
\end{equation}
using $r_j=1$ and that $V(z)$ is even in the first term. 
Next, by the definition of $\mB_j$ \eqref{eq:mBj}, we have, for $k = 1, \dots, \mathcal{N}$ with $k \neq j$,
\begin{equation}
\begin{split}  
\label{eq:mBjmBk} 
(\mB_j-\mB_k)V(a_j-a_k) = &\; \ii \Big( \sum_{l\neq j}^{\cN} r_l \mP_l \alpha(a_j-a_l)-  \sum_{l\neq k}^{\cN} r_l \mP_l \alpha(a_k-a_l)\Bigr)V(a_j-a_k) \\
 = &\; \ii(r_k\mP_k+\mP_j)\alpha(a_j-a_k)V(a_j-a_k) \\ & + \ii\sum_{l\neq j,k}^{\cN} r_l\mP_l V(a_j-a_k)\big( \alpha(a_j-a_l) - \alpha(a_k-a_l) \big), 
\end{split} 
\end{equation} 
using $r_j=1$ and that $\alpha(z)$ is odd to simplify $r_k\mP_k\alpha(a_j-a_k)- r_j\mP_j\alpha(a_k-a_j)=(r_k\mP_k+\mP_j)\alpha(a_j-a_k)$. 
By inserting \eqref{eq:dotPk2} and \eqref{eq:mBjmBk} into \eqref{eq:ddotajfj2} and simplifying, we obtain 
\begin{equation} 
\label{eq:ddotajfj3}
\begin{split} 
\ddot a_j\langle f_j|  = & 2\sum_{k\neq j}^\cN \sum_{l\neq j,k}^{\cN}(r_k+r_l)  \langle f_j| [\mP_k,\mP_l]\alpha(a_j-a_k)V(a_k-a_l) \\
& +2\sum_{k\neq j}^\cN \langle f_j|\big( 2\mP_k+ r_k \{\mP_k,\mP_j\} +(2+r_k) [\mP_k,\mP_j] \big)\alpha(a_j-a_k)V(a_j-a_k) \\ 
& + 2\sum_{k\neq j}^\cN \sum_{l\neq j,k}^{\cN} \langle f_j|\big( r_kr_l \{\mP_k,\mP_l\} +r_l[\mP_k,\mP_l]) 
\big( \alpha(a_j-a_l) - \alpha(a_k-a_l) \big)V(a_j-a_k),   
\end{split} 
\end{equation} 
using $r_k^2=1$, $r_k\{\mP_k,r_k\mP_k+\mP_j\}=2\mP_k^2+r_k\{\mP_k,\mP_j\}$, $[\mP_k,r_k\mP_k+\mP_j]=[\mP_k,\mP_j]$, and 
\begin{equation} 
\label{eq:mPsquare} 
\mP_k^2 = |e_k\rangle\langle f_k|e_k\rangle\langle f_k| =  |e_k\rangle\langle f_k| =  \mP_k
\end{equation}  
by \eqref{eq:sCM3c}. 
Since $\langle f_j|\mP_j = \langle f_j|e_j\rangle\langle f_j|=\langle f_j|$ by \eqref{eq:sCM3c}, we can simplify further: 
\begin{equation}
\label{eq:ddotajfj3A}
 \langle f_j|\big( 2\mP_k+ r_k \{\mP_k,\mP_j\} +(2+r_k) [\mP_k,\mP_j] \big) 
=  2(1+r_k)\langle f_j| \mP_k\mP_j.
\end{equation} 
Moreover, since $V(z)$ is even, 
\begin{multline} \label{eq:ddotajfj3B}
2\sum_{k\neq j}^\cN \sum_{l\neq j,k}^{\cN}(r_k+r_l)  \langle f_j| [\mP_k,\mP_l]\alpha(a_j-a_k)V(a_k-a_l) \\
= 
\sum_{k\neq j}^\cN \sum_{l\neq j,k}^{\cN}(r_k+r_l)  \langle f_j| [\mP_k,\mP_l]\big( \alpha(a_j-a_k)-  \alpha(a_j-a_l)\big)V(a_k-a_l) . 
\end{multline} 
To proceed, we need the identities 
\begin{equation}
\label{eq:IdalphaV}
\alpha(z)V(z)=-\frac12V'(z)
\end{equation}
and  
\begin{equation} 
\label{eq:IdalalV} 
\big( \alpha(a_j-a_l) - \alpha(a_k-a_l) \big)V(a_j-a_k)  
= -\big( \alpha(a_j-a_k) - \alpha(a_j-a_l) \big)V(a_k-a_l).  
\end{equation} 
The first identity \eqref{eq:IdalphaV} can be obtained by differentiating \eqref{eq:IdV} with respect to $z$ while the second identity \eqref{eq:IdalalV} can be obtained by differentiating \eqref{eq:Idmain} with respect to $b$ and setting $a=a_j$, $b=a_k$, and $c=a_l$. Thus, inserting \eqref{eq:ddotajfj3A}, \eqref{eq:ddotajfj3B} with \eqref{eq:IdalalV}, and \eqref{eq:IdalphaV}, \eqref{eq:ddotajfj3} becomes  
\begin{equation} 
\label{eq:ddotajfj4}
\begin{split} 
\ddot a_j\langle f_j|  = & \sum_{k\neq j}^\cN \sum_{l\neq j,k}^{\cN}(r_k-r_l)  \langle f_j| [\mP_k,\mP_l]\big( \alpha(a_j-a_k) - \alpha(a_j-a_l)\bigl)V(a_k-a_l) \\
& -2\sum_{k\neq j}^\cN (1+r_k) \langle f_j| \mP_k\mP_j   V'(a_j-a_k)\\
& - 2\sum_{k\neq j}^\cN \sum_{l\neq j,k}^{\cN} r_kr_l\langle f_j| \{\mP_k,\mP_l\}\big( \alpha(a_j-a_k) - \alpha(a_j-a_l) \big)V(a_k-a_l).
\end{split} 
\end{equation} 
Since $V(z)$ is even, the double sums in the second and third lines in \eqref{eq:ddotajfj4} both vanish by symmetry, and using that $(1+r_k)=2$ for $k=1,\ldots,N$ and $0$ otherwise, 
we get 
\begin{equation} 
 \ddot a_j\langle f_j|  =  - 4\sum_{k\neq j}^N \langle f_j|\mP_k\mP_j V'(a_j-a_k)\quad (j=1,\ldots,N). 
\end{equation} 
By multiplying this from the right with $|e_j\rangle$ and using $\mP_j|e_j\rangle=|e_j\rangle$ and $\langle f_j|\mP_k|e_j\rangle = \langle f_j|e_k\rangle\langle f_k|e_j\rangle$, we obtain \eqref{eq:sCM1a}. 

\section{Multi-soliton solutions of the sBO equation} 
\label{sec:sBO} 
In this section, we present and derive multi-soliton solutions of the sBO equation \eqref{eq:sBO}, both in the real-line case and the  $L$-periodic case. 
The special functions $\alpha(z)$ and $V(z)$ are given in \eqref{eq:alpha} and \eqref{eq:V}, with the real-line case corresponding to I (rational case) and the  $L$-periodic case corresponding to II (trigonometric case), respectively, throughout this section. All our results hold true in both cases. 

\subsection{Result}
\label{sec:sBOresult}
We fix $d\in\Z_{\geq 1}$ and consider the case where $\mU$ is a $d\times d$ matrix-valued function;  
 for $d=1$, \eqref{eq:sBO} reduces to the standard BO equation, and our result below is well-known in this case \cite{chen1979}.  
 
The following theorem, whose proof is given in Section \ref{sec:sBOproof}, is our main result about the sBO equation.

\begin{theorem} 
\label{thm:sBO} 
Let $\{a_j(t),|e_j(t)\rangle,\langle f_j(t)|\}_{j=1}^N$ be a solution of the time evolution equations \eqref{eq:sCM1a}--\eqref{eq:sCM1b} with initial conditions 
\begin{equation} 
a_j(0)=a_{j,0},\quad \dot a_j(0)=v_j,\quad |e_j(0)\rangle =  |e_{j,0}\rangle,\quad \langle f_j(0)|=\langle f_{j,0}|\quad (j=1,\ldots,N)
\end{equation} 
satisfying the constraints 
\begin{equation} 
\label{eq:imajt=0}
\im(a_{j,0})<0 \quad (j=1,\ldots,N), 
\end{equation} 
\begin{equation} 
\label{eq:fjej}
\langle f_{j,0}|e_{j,0}\rangle =1 \quad (j=1,\ldots,N), 
\end{equation} 
and 
\begin{equation}
\label{eq:BTt=0} 
\begin{split} 
v_j \langle f_{j,0}| = &\; 2\ii\sum_{k\neq j}^N  \langle f_{j,0}|e_{k,0}\rangle \langle f_{k,0}|\alpha(a_{j,0}-a_{k,0}) \\
&  -2\ii\sum_{k=1}^N  \langle f_{j,0}|f_{k,0}\rangle\langle e_{k,0}|\alpha(a_{j,0}-a^*_{k,0})\quad (j=1,\ldots,N). 
\end{split} 
\end{equation} 
Then,
\begin{equation} 
\label{eq:ansatz} 
\mU(x,t)= \ii\sum_{j=1}^N |e_j(t)\rangle\langle f_j(t)|\alpha(x-a_j(t)) -  \ii\sum_{j=1}^N |f_j(t)\rangle\langle e_j(t)|\alpha(x-a^*_j(t))  
\end{equation} 
is a solution of the sBO equation \eqref{eq:sBO} for all times $t$ provided that the following condition holds true,
\begin{equation} 
\label{eq:imaj}
\im(a_j(t))<0 \quad (j=1,\ldots,N).
\end{equation} 
\end{theorem} 

\begin{remark} 
We expect that the condition in \eqref{eq:imaj} is automatically fulfilled for all times $t\in\R$ under the stated assumptions, and thus can be dropped. 
It would be interesting to prove, or falsify, this expectation. Similar remarks apply to Theorems~\ref{thm:sncILW}, \ref{thm:sBOgen} and \ref{thm:sncILWgen}.
\end{remark} 

\subsubsection{One-soliton solutions}
For $N=1$, the time evolution equations \eqref{eq:sCM1a}--\eqref{eq:sCM1b} simplify to $\ddot a_1=0$, $|\dot e_1\rangle=0$, and $\langle \dot f_1|=0$. Moreover, the general solution of the constraints \eqref{eq:imajt=0}--\eqref{eq:BTt=0} is $|e_{1,0}\rangle = |f_{1,0}\rangle/\langle f_{1,0}|f_{1,0}\rangle$ with $\langle f_{1,0}|\in\cV^*$ an arbitrary non-zero vector, together with
\begin{equation} 
v_1 = -2\ii\alpha(a_{1,0}^{\phantom *}-a_{1,0}^*) = \begin{cases} -1/\aI_1 & \text{(I)}\\ - (2\pi/L)\coth(2\pi\aI_1/L) & \text{(II)}  \end{cases} ,  
\end{equation} 
where $a_{1,0}=\aR_1+\ii\aI_1$ with $\aR_1\in\R$ and $\aI_1<0$. Thus, $a_1(t)=\aR_1+\ii\aI_1+v_1 t$, $|e_1(t)\rangle\langle f_1(t)|=|f_{1,0}\rangle\langle f_{1,0}|/\langle f_{1,0}|f_{1,0}\rangle$ independent of $t$,
and Theorem~\ref{thm:sBO} implies the following explicit formula for the one-soliton solutions of the sBO equation \eqref{eq:sBO}, 
\begin{equation} 
\label{eq:onesoliton} 
\mU(x,t)= \frac{|f_{1,0}\rangle\langle f_{1,0}|}{\langle f_{1,0}|f_{1,0}\rangle}
\big(\ii\alpha(x-\aR_1-\ii\aI_1 - v_1t) - \ii\alpha(x-\aR_1+\ii\aI_1 - v_1t) \big) .
\end{equation} 
Thus, a soliton of the sBO equation is characterized by $d$ complex parameters: the pole $a_{j,0}$ at time $t=0$ (one complex parameter), and a non-zero vector $\langle f_{1,0}|\in\cV^*$ modulo the transformations $\langle f_{1,0}| \to c\langle f_{1,0}|$, $c\in\C\setminus\{0\}$ arbitrary, which leave \eqref{eq:onesoliton} invariant ($d-1$ complex parameters).
 It is interesting to note that only $v_1>0$ is allowed: the sBO equation is chiral in the sense that its solitons can only move to the right (note that the center of the soliton is located at the position $x=\aR_1+v_1t$).
 
\subsubsection{Non-interacting multi-soliton solutions}
Let $N\leq d$, and pick $N$ non-zero orthogonal vectors $\langle f_{j,0}|\in\cV^*$: 
\begin{equation} 
\langle f_{j,0}|f_{k,0}\rangle=\delta_{j,k}\langle f_{j,0}|f_{j,0}\rangle\quad (j,k=1,\ldots,N). 
\end{equation} 
One can check that $a_j(t)=\aR_j+\ii\aI_j+v_j t$ and $|e_j(t)\rangle\langle f_j(t)|=|f_{j,0}\rangle\langle f_{j,0}|/\langle f_{j,0}|f_{j,0}\rangle$
is a solution of the time evolution equations \eqref{eq:sCM1a}--\eqref{eq:sCM1b} satisfying the constraints \eqref{eq:BTt=0}--\eqref{eq:imaj} provided that $\aR_j\in\R$, $\aI_j<0$ and 
 \begin{equation} 
v_j = \begin{cases} -1/\aI_j & \text{(I)}\\ - (2\pi/L)\coth(2\pi\aI_j/L) & \text{(II)}  \end{cases} 
\end{equation} 
for $j=1,\ldots,N$. Thus, by Theorem~\ref{eq:sBO}, 
\begin{equation} 
\mU(x,t)= \sum_{j=1}^N \frac{|f_{j,0}\rangle\langle f_{j,0}|}{\langle f_{j,0}|f_{j,0}\rangle}
\big(\ii\alpha(x-\aR_j-\ii\aI_j - v_jt) - \ii\alpha(x-\aR_j+\ii\aI_j - v_jt) \big) 
\end{equation} 
is an exact solution of the sBO equation \eqref{eq:sBO}. Clearly, this solution describes $N$ non-interacting one-solitons. 

\subsubsection{Generic multi-soliton solutions}
\label{sec:NsolitonssBO}
The initial data for the $N$-soliton in Theorem \ref{thm:sBO} is specified in terms of the $2N$ complex numbers $a_{j,0}$ and $v_j$,
as well as the $2N$ vectors $|e_{j,0}\rangle \in \mathcal{V}$ and $\langle f_{j,0}| \in \mathcal{V}^*$ ($j=1,\ldots,N$). However, the constraints \eqref{eq:fjej} and \eqref{eq:BTt=0} imply that these parameters cannot all be independently specified. Moreover, some choices of these parameters give rise to the same soliton solution $\mathsf{U}$. In what follows, we show that, in the generic case when certain matrices are invertible, the constraints \eqref{eq:fjej} and \eqref{eq:BTt=0} can be solved explicitly in terms of only linear operations. As a result, we will see that if $\mathcal{V}$ is (complex) $d$-dimensional, then the family of $N$-solitons of the sBO equation generically depends on $Nd$ complex parameters. Moreover, we give a recipe for the construction of the $N$-soliton solution in terms of these $Nd$ parameters.

For each $j = 1, \dots, N$, let us identify $|e_{j,0}\rangle \in \mathcal{V}$ with the vector $\mathbf{e}_j \in \C^d$ whose components $(\mathbf{e}_j)_{\mu}$, ${\mu} = 1, \dots, d$, are the components of $|e_{j,0}\rangle$ with respect to some given basis of $\mathcal{V}$.
Next, let us identify the collection of $N$ vectors $\mathbf{e}_j$, $j=1,\ldots,N$, with the single vector $\mathbf{e} \in \C^{Nd}$ whose components $\mathbf{e}_{j,\mu} := (\mathbf{e}_j)_{\mu}$ are indexed by $j=1,\ldots,N$ and ${\mu}=1,\ldots,d$.
Similarly, let us identify the three collections of vectors $\langle e_{j,0}|$, $\langle f_{j,0}|$, and $|f_{j,0}\rangle$, $j=1,\ldots,N$, with the vectors $\ve^*=(e^*_{j,\mu})\in \C^{Nd}$, $\vf^*=(f^*_{j,\mu}) \in \C^{Nd}$, and $\vf=(f_{j,\mu})\in \C^{Nd}$, respectively. (Since all considerations here will concern the constraints at time $t = 0$, we do not include the subscript $0$ in our notation for simplicity.)
With this notation, we can write the constraint \eqref{eq:BTt=0} as 
\begin{equation} 
\label{eq:system}
\mA\ve^*+ \mB\ve  = \mC \vf^*  
\end{equation} 
where the $Nd \times Nd$ matrices $\mA,\mB$, and $\mC$ are given by 
\begin{equation} 
\label{ABdvdef}
\begin{split} 
A_{j,\mu;k,\nu} &= -2\ii  \langle f_j|f_k\rangle\delta_{\mu,\nu}\alpha(a^{\phantom*}_{j,0}-a_{k,0}^*),\\
B_{j,\mu;k,\nu} &= 2\ii (1-\delta_{j,k})f^*_{k,{\mu}}f^*_{j,{\nu}}\alpha(a_{j,0}-a_{k,0}), \\
C_{j,\mu;k,\nu} &= v_j\delta_{j,k}\delta_{\mu,\nu}. 
\end{split} 
\end{equation} 
We write \eqref{eq:system} and the negative of its complex conjugate as the linear system,\footnote{Note that the star in $\mA^*$ etc.\ means complex conjugation, i.e., $\mA^*$ is given by the matrix elements $(A_{j,\mu;k,\nu})^*$ where $A_{j,\mu;k,\nu}$ are the matrix elements of $\mA$.} 
\begin{equation} 
\label{eq:system1}
 \begin{pmatrix} \mA & \mB \\ -\mB^* & -\mA^* \end{pmatrix}\begin{pmatrix}
\ve^*\\ \ve \end{pmatrix} = \begin{pmatrix} \mC\vf^*\\ -\mC^* \vf \end{pmatrix}, 
\end{equation} 
and note that the $2Nd\times 2Nd$ matrix in \eqref{eq:system1} is hermitian. Thus, restricting ourselves to the generic case when this $2Nd\times 2Nd$ matrix is invertible, we obtain
\begin{equation} 
\label{eq:solve} 
\begin{pmatrix}
\ve^*\\ \ve \end{pmatrix} 
= \begin{pmatrix} \mA & \mB \\ -\mB^* & -\mA^* \end{pmatrix}^{-1} 
\begin{pmatrix} \mC\vf^*\\ - \mC^* \vf \end{pmatrix}. 
\end{equation} 
Substitution of these expressions for $\mathbf{e}$ and $\mathbf{e}^*$ into the constraint \eqref{eq:fjej} and its complex conjugate gives $2N$ linear equations which can be solved uniquely for the $2N$ initial velocities $v_j, v_j^* \in \C$, $j = 1, \dots, N$, provided that the relevant determinant is nonzero (this is the generic case).
Given these expressions for $v_j$, the vectors $\mathbf{e}$ and $\mathbf{e}^*$ can be found from (\ref{eq:solve}). This means that all the constraints can be solved in terms of $\{a_{j,0}\}_{j=1}^N$ and $\mathbf{f}$.
So in the Hermitian case considered here, the class of $N$-soliton solutions of the sBO equation is parametrized by the $N$ complex parameters $\{a_{j,0}\}_{j=1}^N$ as well as $Nd$ further complex parameters needed to determine $\mathbf{f}$. However, some of these parameter configurations yield the same $N$-soliton solution $\mathsf{U}(x,t)$ via \eqref{eq:ansatz}. Indeed, the equations \eqref{eq:sCM1a}--\eqref{eq:sCM1b} and \eqref{eq:imajt=0}--\eqref{eq:ansatz} are invariant under the replacements
\begin{equation} 
|e_j\rangle \to c_j |e_j\rangle, \qquad \langle f_j | \to \frac{1}{c_j} \langle f_j |,
\end{equation} 
where $\{c_j\}_{j=1}^N$ are nonzero complex constants. This means that the family of $N$-solitons generically is $Nd$ complex dimensional; in other words, each soliton is specified by $d$ complex parameters, which is consistent with the result for the one-soliton obtained above.

\subsection{Proof of Theorem~\ref{thm:sBO}}
\label{sec:sBOproof}
We start with the spin-pole ansatz 
\begin{equation} 
\label{eq:ansatzh} 
\mU(x,t)=\ii \sum_{j=1}^N \mP_j(t)\alpha(x-a_j(t)) - \ii \sum_{j=1}^N \mP^\dag_j(t)\alpha(x-a^*_j(t))  
\end{equation}  
where $\mP_j=\mP_j(t)$ are $d\times d$ matrices and $a_j=a_j(t)\in\C$; the $\mP_j$ and $a_j$ correspond to the spin degrees of freedom and poles, respectively. 

\begin{proposition} 
\label{prop:sBO1} 
The function $\mU(x,t)$ in \eqref{eq:ansatzh} satisfies the sBO equation \eqref{eq:sBO} provided \eqref{eq:imaj} and the following equations hold true,
\begin{equation} 
\label{eq:BT1a} 
\dot a_{j} \mP_{j} = 2\ii\sum_{k\neq j}^N  \mP_{j}\mP_{k}\alpha(a_{j}-a_{k}) -2\ii\sum_{k=1}^N \mP_{j}\mP^\dag _{k}\alpha(a_{j}-a^*_{k})\quad (j=1,\ldots,N), \\
\end{equation}
\begin{equation} 
\label{eq:BT2a} 
\dot \mP_j  = -2\ii\sum_{k \neq j}^N[\mP_j,\mP_k]V(a_j-a_k) \quad (j=1,\ldots,N),
\end{equation} 
and 
\begin{equation} 
\label{eq:BT3a} 
\mP_j^2=\mP_j  \quad (j=1,\ldots, N).
\end{equation} 
\end{proposition} 

\begin{proof} 
We use the short-hand notation
\begin{equation} 
\label{eq:shorthand} 
(a_j,\mP_j,r_j)\coloneqq \begin{cases} (a_j,\mP_j,+1) & (j=1,\ldots,N) \\ (a^*_{j-N},\mP^\dag_{j-N},-1) & (j=N+1,\ldots,\cN) \end{cases} , \quad \cN=2N
\end{equation} 
to write \eqref{eq:ansatzh} as  
\begin{equation} 
\label{eq:U} 
\mU= \ii\sum_{j=1}^{\cN} r_j \mP_j\alpha(x-a_j) .
\end{equation} 
The proof will follow by inserting \eqref{eq:U} into the sBO equation \eqref{eq:sBO} and performing long but straightforward computations.

We compute each term in \eqref{eq:sBO}. We start with 
\begin{equation} 
\label{eq:Udot}
\mU_t =  \sum_{j=1}^\cN \Big( \ii r_j \dot \mP_j\alpha(x-a_j) - \ii r_j \mP_j \dot a_j \alpha'(x-a_j)\Big)  . 
\end{equation} 
Next, we compute
\begin{align}
\label{eq:UUx1}
\{\mU,\mU_x\} =&\; -\sum_{j=1}^{\cN}\sum_{k=1}^{\cN}  r_j r_k \{\mP_j, \mP_k\} \alpha(x-a_j)\alpha'(x-a_k)            \nonumber  \\
=&\; -2\sum_{j=1}^{\cN}  \mP_j^2\alpha(x-a_j)\alpha'(x-a_j)  -\sum_{j=1}^{\cN}\sum_{k\neq j}^{\cN} r_jr_k \{\mP_j,\mP_k\} \alpha(x-a_j)\alpha'(x-a_k)   \nonumber \\
=&\; \sum_{j=1}^{\cN} \mP_j^2 \alpha''(x-a_j) + \sum_{j=1}^{\cN}\sum_{k\neq j}^{\cN} r_j r_k \{\mP_j,\mP_k\}\alpha(a_j-a_k)\alpha'(x-a_k) \nonumber \\
&\; +\sum_{j=1}^{\cN}\sum_{k\neq j}^{\cN}r_j r_k \{\mP_j,\mP_k\}V(a_j-a_k)\big(\alpha(x-a_j)-\alpha(x-a_k)\big),
\end{align}
where we have used the identities 
\begin{equation} 
\label{eq:Id1} 
2\alpha(x-a_j)\alpha'(x-a_j)=-\alpha''(x-a_j)
\end{equation} 
and\footnote{We write the following identity in a seemingly strange way, mixing $\alpha'$ and $V=-\alpha'$, to emphasize the similarity with a corresponding identity \eqref{eq:Aidentity2} used in the sncILW case.} 
\begin{equation}
\label{eq:Id2}  
\alpha(x-a_j)\alpha'(x-a_k)=- \alpha(a_j-a_k)\alpha'(x-a_k) -V(a_j-a_k)\big(\alpha(x-a_j)-\alpha(x-a_k)\big).
\end{equation} 
The first identity \eqref{eq:Id1} can be obtained by differentiating \eqref{eq:IdV} with respect to $z$ and setting $z=x-a_j$ while the second identity \eqref{eq:Id2} can be obtained by differentiating \eqref{eq:Idmain} with respect to $c$ and setting $a=x$, $b=a_j$, and $c=a_k$.

The final sum in \eqref{eq:UUx1} vanishes because the summand is antisymmetric under the interchange $j\leftrightarrow k$. Hence,  after re-labelling summation indices  $j\leftrightarrow k$ in the second sum in \eqref{eq:UUx1} using \eqref{eq:Id4}, we are left with
\begin{equation}
\label{eq:UUx2}
\{\mU,\mU_x\} = \sum_{j=1}^{\cN} \mP_j^2 \alpha''(x-a_j)-\sum_{j=1}^{\cN}\sum_{k\neq j}^{\cN} r_jr_k \{\mP_j,\mP_k\}\alpha(a_j-a_k)\alpha'(x-a_j).
\end{equation}

To compute terms in \eqref{eq:sBO} involving the Hilbert transform $H$, we use
\begin{equation}
\label{eq:Id3} 
\begin{split}
(H\alpha'(\cdot-a_j))(x) = \ii r_j\alpha'(x-a_j) \\
\end{split} 
\end{equation}
(this follows from the well-known facts (i) $\alpha(x-a)$ is an eigenfunction of $H$ with eigenvalue $+\ii$ for $\im(a)<0$ and  $-\ii$ for $\im(a)>0$ \cite{chen1979} and (ii) $H$ commutes with differentiation \cite[Chapter~4.8]{king2009}). Hence,
\begin{equation}
\label{eq:HUx}
H\mU_x= -\sum_{j=1}^{\cN} \mP_j \alpha'(x-a_j),\qquad H\mU_{xx}= -\sum_{j=1}^{\cN} \mP_j \alpha''(x-a_j),
\end{equation}
where we have again used the fact that $H$ commutes with differentiation to derive the second equation from the first.
From \eqref{eq:U} and the first equation in \eqref{eq:HUx}, we compute
\begin{equation}
\label{eq:UHUx1}
\begin{split} 
\ii [\mU,H\mU_x]=&\; \sum_{j=1}^{\cN}\sum_{k\neq j}^{\cN} r_j [\mP_j,\mP_k] \alpha(x-a_j)\alpha'(x-a_k) \\
=&\; -\sum_{j=1}^{\cN} \sum_{k\neq j}^{\cN} r_j[\mP_j,\mP_k] \alpha(a_j-a_k)\alpha'(x-a_k)\\
 & -\sum_{j=1}^{\cN}\sum_{k\neq j}^{\cN} r_j[\mP_j,\mP_k] V(a_j-a_k)\big(\alpha(x-a_j)-\alpha(x-a_k)\big), 
\end{split} 
\end{equation}
inserting \eqref{eq:Id2} 
in the second step. We can rewrite the second sum as follows, 
\begin{multline}
\sum_{j=1}^{\cN} \sum_{k\neq j}^{\cN} r_j[\mP_j,\mP_k] V(a_j-a_k)\big(\alpha(x-a_j)-\alpha(x-a_k)\big) \\
=  \frac12 \sum_{j=1}^{\cN} \sum_{k\neq j}^{\cN} (r_j+r_k)[\mP_j,\mP_k] V(a_j-a_k)\big(\alpha(x-a_j)-\alpha(x-a_k)\big) \\
=\sum_{j=1}^{\cN} \sum_{k\neq j}^{\cN} (r_j+r_k)[\mP_j,\mP_k] V(a_j-a_k) \alpha(x-a_j)
\end{multline}
since $V(z)$ is an even function.  
Also changing variables $j\leftrightarrow k$ in the first sum in \eqref{eq:UHUx1} using that $\alpha(z)$ is odd, we arrive at
\begin{equation}
\label{eq:UHUx2}
\begin{split} 
\ii [\mU,H\mU_x]=&\; -\sum_{j=1}^{\cN}\sum_{k\neq j}^{\cN} r_k[\mP_j,\mP_k] \alpha(a_j-a_k)\alpha'(x-a_j)
\\ & -\sum_{j=1}^{\cN}\sum_{k\neq j}^{\cN} (r_j+r_k)[\mP_j,\mP_k]V(a_j-a_k)\alpha(x-a_j).
\end{split} 
\end{equation}
Inserting \eqref{eq:Udot}, \eqref{eq:UUx2}, the second equation in \eqref{eq:HUx}, and \eqref{eq:UHUx2} into \eqref{eq:sBO} gives
\begin{equation} 
\begin{split}
0=&\; \sum_{j=1}^{\cN} \Bigg(\ii r_j \dot{\mP}_j- \sum_{k\neq j}^{\cN} (r_j+r_k)[\mP_j,\mP_k]V(a_j-a_k)\Bigg) \alpha(x-a_j) \\
&\; +\sum_{j=1}^{\cN} \Bigg(-\ii r_j \mP_j\dot{a}_j -\sum_{k\neq j}^{\cN} r_k \big( r_j\{\mP_j,\mP_k\}    +[\mP_j,\mP_k]  \big)\alpha(a_j-a_k)\Bigg)\alpha'(x-a_j) \\
&\;  + \sum_{j=1}^{\cN} \big( \mP_j^2 -\mP_j \big) \alpha''(x-a_j).
\end{split}
\end{equation} 
Thus, the function $\mU$ defined in \eqref{eq:U} satisfies the sBO equation \eqref{eq:sBO} if and only if the following conditions are fulfilled,  
\begin{equation} 
\label{eq:BT1} 
\mP_j\dot a_j =  \ii \sum_{k\neq j}^{\cN} r_k \bigl(  \{\mP_j,\mP_k\} + r_j[\mP_j,\mP_k]  \big) \alpha(a_j-a_k), \\
\end{equation} 
\begin{equation} 
\label{eq:BT2} 
\dot \mP_j =  -\ii \sum_{k\neq j}^{\cN}(1+r_jr_k)[\mP_j,\mP_k]V(a_j-a_k) ,\\
\end{equation} 
\begin{equation} 
\label{eq:BT3} 
\mP_j^2 = \mP_j 
\end{equation} 
for $j=1,\ldots,\cN$. Recalling \eqref{eq:shorthand}, one can check that \eqref{eq:BT1}, \eqref{eq:BT2}, and \eqref{eq:BT3} are equivalent to \eqref{eq:BT1a}, \eqref{eq:BT2a} and \eqref{eq:BT3a}, respectively. 
\end{proof} 

We now make the ansatz 
\begin{equation}
\label{eq:mPj} 
\mP_j= |e_j\rangle\langle f_j| \quad (j=1,\ldots,N) 
\end{equation} 
with vectors $|e_j\rangle\in\cV$ and $\langle f_j|\in\cV^*$. 
Then the function $\mU(x,t)$ defined in \eqref{eq:ansatzh} satisfies the equations in Proposition~\ref{prop:sBO1} whenever $\{a_j,|e_j\rangle,\langle f_j|\}_{j=1}^N$ satisfy the equations defining the B\"acklund transformations of the sCM system discussed in Section~\ref{sec:sCM_BT}. The precise statement is as follows.

\begin{lemma}
\label{lem:sBOlemma}
Suppose that $\{a_j,|e_j\rangle,\langle f_j|\}_{j=1}^N$ satisfy the equations \eqref{eq:sCM1b}, \eqref{eq:sCM1c} and \eqref{eq:BThermitian}, and $\mP_j$ is given by \eqref{eq:mPj}. 
Then $\mU$ \eqref{eq:ansatzh} satisfies \eqref{eq:BT1a}--\eqref{eq:BT3a}.  
\end{lemma} 

\begin{proof} 
By multiplying \eqref{eq:BThermitian} from the left by $|e_j\rangle$, one gets 
\begin{equation*}
\dot a_j |e_j\rangle\langle f_j| = 2\ii\sum_{k\neq j}^N  |e_j\rangle\langle f_j|e_k\rangle \langle f_k|\alpha(a_j-a_k) -2\ii\sum_{k=1}^M  |e_j\rangle\langle f_j|f_k\rangle\langle e_k|\alpha(a_j-a_k^*)
\end{equation*} 
which is \eqref{eq:BT1a} (note that $\mP_j^\dag = |f_j\rangle\langle e_j|$). 

Using \eqref{eq:sCM1b}, we compute 
\begin{equation} 
\begin{split} 
\dot\mP_j = & |e_j\rangle\langle \dot f_j|+|\dot e_j\rangle\langle f_j|  \\
=  & -2\ii\sum_{k\neq j}^N \big(  |e_j\rangle\langle f_j|e_k\rangle \langle f_k| -  |e_k\rangle\langle f_k|e_j\rangle\langle f_j|  \big) 
V(a_j-a_k) \\
= &  -2\ii\sum_{k\neq j}^N \big( \mP_j\mP_k-\mP_k\mP_j \big) V(a_j-a_k) , 
\end{split} 
\end{equation}    
which is \eqref{eq:BT2a}.

Finally, using \eqref{eq:sCM1c}, 
\begin{equation} 
\mP_j^2 = |e_j\rangle\langle f_j|e_j\rangle\langle f_j| =  |e_j\rangle\langle f_j|=\mP_j, 
\end{equation}  
which is \eqref{eq:BT3a}. 
\end{proof} 

Then, the theorem is implied by Proposition~\ref{prop:BT}.

\section{Multi-soliton solutions of the sncILW equation}
\label{sec:sncILW}
In this section, we present results for the sncILW equation on the real line in analogy with those for the sBO equation in the previous section. 
We use the hyperbolic case (III) special functions $\alpha(z)=(\pi/2\delta)\coth(\pi z/2\delta)$ and 
$V(z)=(\pi/2\delta)^2/\sinh^2(\pi z/2\delta)$ with $\delta>0$; see \eqref{eq:alpha} and \eqref{eq:V}.

\subsection{Result}
We use the same conventions as described in Section~\ref{sec:sBOresult}. In the case $d=1$, \eqref{eq:sncILW} reduces to the standard ncILW equation, whose soliton solutions were constructed in \cite{berntson2020}. Our main result about the sncILW equation, stated below, generalizes the construction in \cite{berntson2020}.

\begin{theorem} 
\label{thm:sncILW} 
Let $\{a_j(t),|e_j(t)\rangle,\langle f_j(t)|\}_{j=1}^N$ be a solution of the time evolution equations \eqref{eq:sCM1a}--\eqref{eq:sCM1b} with initial conditions 
\begin{equation} 
a_j(0)=a_{j,0},\quad \dot a_j(0)=v_j,\quad |e_j(0)\rangle =  |e_{j,0}\rangle,\quad \langle f_j(0)|=\langle f_{j,0}|\quad (j=1,\ldots,N)
\end{equation} 
satisfying the constraints 
\begin{equation} 
\label{eq:imajtsncILW=0}
-\frac{3\delta}{2}<\im(a_{j,0})< -\frac{\delta}{2}   \quad (j=1,\ldots,N), 
\end{equation} 
\begin{equation} 
\label{eq:fjejsncILW}
\langle f_{j,0}|e_{j,0}\rangle =1 \quad (j=1,\ldots,N), 
\end{equation} 
and 
\begin{equation}
\label{eq:BTt=0_sncILW} 
\begin{split} 
v_j \langle f_{j,0}| = & 2\ii\sum_{k\neq j}^N  \langle f_{j,0}|e_{k,0}\rangle \langle f_{k,0}|\alpha(a_{j,0}-a_{k,0}) \\
&  -2\ii\sum_{k=1}^N  \langle f_{j,0}|f_{k,0}\rangle\langle e_{k,0}|\alpha(a_{j,0}-a^*_{k,0}+\ii\delta)\quad (j=1,\ldots,N). 
\end{split} 
\end{equation} 
Then,
\begin{equation} 
\label{eq:mUmV} 
\begin{split}
\mU(x,t)=&\;  \ii\sum_{j=1}^N |e_j(t)\rangle\langle f_j(t)|\alpha(x-a_j(t) -\ii\delta/2) -  \ii\sum_{j=1}^N |f_j(t)\rangle\langle e_j(t)|\alpha(x-a^*_j(t)+\ii\delta/2),  \\
\mV(x,t)=&\;  -\ii\sum_{j=1}^N |e_j(t)\rangle\langle f_j(t)|\alpha(x-a_j(t)+\ii\delta/2) +  \ii\sum_{j=1}^N |f_j(t)\rangle\langle e_j(t)|\alpha(x-a^*_j(t)-\ii\delta/2) 
\end{split}
\end{equation} 
is a solution of the sncILW equation \eqref{eq:sncILW} for all times $t$ provided that the following condition holds true,
\begin{equation} 
\label{eq:imaj_sncILW}
-\frac{3\delta}{2}<\im( a_j(t)) <-\frac{\delta}{2} \quad (j=1,\ldots,N).
\end{equation} 
\end{theorem} 

\subsection{Proof of Theorem~\ref{thm:sncILW}}
Since details of the proof are very similar to the ones of Theorem~\ref{thm:sBO}, we only explain the key differences. 

The proof is facilitated by introducing the notation
\begin{equation}
\label{eq:circ} 
\left(\begin{array}{c} F_1 \\ F_2  \end{array}\right)\circ \left(\begin{array}{c} G_1 \\ G_2 \end{array}\right)\coloneqq  \left(\begin{array}{c}F_1G_1 \\ -F_2G_2 \end{array}\right)
\end{equation}
for $\C$-valued functions $F_j$, $G_j$ ($j=1,2$), and the operator
\begin{equation}
\label{eq:cT} 
\cT\coloneqq \left(\begin{array}{cc} T & \tilde{T} \\ -\tilde{T} & -T     \end{array}\right)
\end{equation}
to be interpreted as linear operator acting on vector-valued functions, see \cite{berntson2020}. 
In the present paper, we use the product $\circ$ defined in  \eqref{eq:circ} 
also for vectors $\cF$, $\cG$ whose components $F_j$, $G_j$ are complex $d\times d$ matrices, and we let
\begin{equation}
\{\cF\ocomma \cG\}\coloneqq \cF\circ\cG+\cG\circ \cF,\qquad [\cF\ocomma \cG]\coloneqq \cF\circ\cG- \cG\circ \cF 
\end{equation}
be the corresponding generalizations of the commutator and anticommutator, respectively. 

Using this notation, the system \eqref{eq:sncILW} can be written in terms of the vector
\begin{equation}
\label{eq:cU} 
\cU(x,t)\coloneqq \left(\begin{array}{c}  \mU(x,t) \\ \mV(x,t) \end{array}\right)
\end{equation}
as
\begin{equation}
\label{eq:sncILW2}
\cU_t+ \{\cU\ocomma \cU_x\} +\mathcal{T}\cU_{xx}+\ii [\cU\ocomma \mathcal{T}\cU_x]=0.
\end{equation}
To solve \eqref{eq:sncILW2}, we use the spin-pole ansatz
\begin{equation}
\label{eq:ansatzh_sncILW}
\cU(x,t)=\ii \sum_{j=1}^{N} \mP_j(t) \cA_+(x-a_j(t))-\ii\sum_{j=1}^N \mP_j^{\dag}(t)\cA_-(x-a_j^*(t)),
\end{equation}
where $\mP_j=\mP_j(t)$ are $d\times d$ matrices, $a_j=a_j(t)\in\C$, and
\begin{equation}
\label{eq:cA}
\mathcal{A}_{\pm}(z)\coloneqq \left(\begin{array}{c} +\alpha(z \mp \ii\delta/2) \\ -\alpha(z\pm\ii\delta/2) \end{array}\right)
\end{equation}
for all $z\in\C$ such that $\delta/2<|\im(z)|<3\delta/2$. 
Using the shorthand notation \eqref{eq:shorthand}, we can write \eqref{eq:ansatzh_sncILW} as
\begin{equation}
\label{eq:ansatz_sncILW}
\cU=\ii\sum_{j=1}^{\cN} r_j \mP_j \cA_{r_j}(x-a_j).
\end{equation}
With this notation in place, one can prove Theorem~\ref{thm:sncILW} by the very same computations we presented in our proof of Theorem~\ref{thm:sBO} in Section~\ref{sec:sBOproof}. The reason for this is that the vector-valued function $\cA(z)$ defined in \eqref{eq:cA} satisfies the following identities which are natural generalizations of the identities \eqref{eq:Id1}, \eqref{eq:Id2} and \eqref{eq:Id3}, respectively ($j,k=1,\ldots,\cN$ and $x\in\R$):  
\begin{equation}
\label{eq:Aidentity1}
2\cA_{r_j}(x-a_j)\circ \cA_{r_j}'(x-a_j)=- \cA_{r_j}''(x-a_j) 
\end{equation}
(this is implied by \eqref{eq:Id1} and definitions),
\begin{align}
\label{eq:Aidentity2}
\cA_{r_j}(x-a_j) & \circ \cA_{r_k}'(x-a_k)= -\alpha(a_j-a_k+\ii(r_j-r_k)\delta/2) \cA_{r_k}'(x-a_k) \nonumber \\
&\; -V(a_j-a_k+\ii(r_j-r_k)\delta/2)\big(\cA_{r_j}(x-a_j)-\cA_{r_k}(x-a_k)\big) 
\end{align}
(this is implied by \eqref{eq:Id2} and definitions, together with $\alpha'(z)=-V(z)$ and the fact that $\alpha(z)$ and $V(z)$ both are $2\ii\delta$-periodic), 
and
\begin{equation} 
\label{eq:Aidentity3}
(\cT\cA_{r_j}'(\cdot -a_j))(x)=\ii r_j\cA_{r_j}'(x-a_j) 
\end{equation} 
(this follows from the result, proved in \cite{berntson2020}, that for the vector-valued functions $\cA_\pm(z)$ in \eqref{eq:cA}, $\cA_+'(x-a)$  is an eigenfunction of the matrix operator $\cT$ with eigenvalue $+\ii$ provided that $-3\delta/2 < \im(a) < -\delta/2$, and $\cA_-'(x-a)$ is an eigenfunction of $\cT$ with eigenvalue $-\ii$ provided that $\delta/2 < \im(a) < 3\delta/2$; see \cite[Appendix A.b]{berntson2020}).

The interested reader can find full details on how to prove Theorem~\ref{thm:sncILW} in Appendix~\ref{app:sncILWsolutions}. There, we prove a more general result for solutions that are not necessarily hermitian. The hermitian case of Theorem~\ref{thm:sncILW} can be established by specializing to \eqref{eq:ansatzh} at all points in the proof. 

\section{Further results} 
\label{sec:results}
We briefly summarize further basic results about the sBO and sncILW equations. 

\subsection{Hamiltonian structure}
\label{sec:Hamiltonian}
The sBO equation \eqref{eq:sBO} can be obtained as a Hamiltonian equation $\mU_t=\{\cH_{\rm sBO},\mU\}_{\mathrm{P.B.}}$ from the Hamiltonian ($\mU$ below is short for $\mU(x)$)
\begin{equation} 
\cH_{\rm sBO} = \int \tr\left( \frac13 \mU^3 + \frac12 \mU H\mU_x \right)\,\dd{x} 
\end{equation} 
and the Poisson brackets\footnote{The subscript $\mathrm{P.B.}$ is used to avoid confusion with the anti-commutator.} 
\begin{equation} 
\{ U_{{\mu},{\nu}}(x), U_{{\mu}',{\nu}'}(x')\}_{\mathrm{P.B.}} = \ii\delta(x-x')\big(\delta_{{\nu},{\mu}'}U_{{\mu},{\nu}'}(x)-\delta_{{\mu},{\nu}'}U_{{\mu}',{\nu}}(x)\big) + \delta'(x-x')\delta_{{\nu},{\mu}'}\delta_{{\mu},{\nu}'}  
\end{equation} 
where $U_{{\mu},{\nu}}(x)$ are the matrix elements of $\mU(x)$ and ${\mu},{\nu},{\mu}',{\nu}'=1,\ldots,d$; 
the integration is over $\R$ and $[-L/2,L/2]$ in the real-line and periodic cases, respectively, and $\tr$ is the usual matrix trace (sum of diagonal elements). 

Similarly, the sncILW equation \eqref{eq:sncILW} arises from the Hamiltonian 
\begin{equation} 
\cH_{\rm sncILW} = \int \tr\left( \frac13 \mU^3 +\frac13\mV^3 + \frac12 \mU T\mU_x +  \frac12 \mV T\mV_x + \frac12 \mV \tT\mU_x + \frac12 \mU \tT\mV_x \right)\,\dd{x} 
\end{equation} 
and the Poisson brackets 
\begin{equation} 
\begin{split} 
\{ U_{{\mu},{\nu}}(x), U_{{\mu}',{\nu}'}(x')\}_{\mathrm{P.B.}} &= \ii\delta(x-x')\big(\delta_{{\nu},{\mu}'}U_{{\mu},{\nu}'}(x)-\delta_{{\mu},{\nu}'}U_{{\mu}',{\nu}}(x)\big) + \delta'(x-x')\delta_{{\nu},{\mu}'}\delta_{{\mu},{\nu}'}  ,\\
\{ V_{{\mu},{\nu}}(x), V_{{\mu}',{\nu}'}(x')\}_{\mathrm{P.B.}} &= \ii\delta(x-x')\big(\delta_{{\nu},{\mu}'}V_{{\mu},{\nu}'}(x)-\delta_{{\mu},{\nu}'}V_{{\mu}',{\nu}}(x)\big) - \delta'(x-x')\delta_{{\nu},{\mu}'}\delta_{{\mu},{\nu}'},\\
\{ U_{{\mu},{\nu}}(x), V_{{\mu}',{\nu}'}(x')\}_{\mathrm{P.B.}} &= 0
\end{split} 
\end{equation} 
in the real-line and periodic cases. 

(The verifications of these results are by straightforward computations which we omit.) 

\subsection{Local limits of the sILW equation}
\label{sec:locallimit}
We give details on how the matrix KdV equation \eqref{eq:mKdV} and the (generalization of the) HF equation \eqref{eq:sHF} are obtained as limits $\delta\downarrow 0$ from the sILW equation \eqref{eq:sILW}; note that, while the latter equation is non-local, the former equations both are local.

For simplicity, we only consider the real-line case with $T$ given by \eqref{eq:TT}. 
We recall that
\begin{equation} 
\label{eq:Texpand}
(Tf_x)(x) = -\frac1\delta f(x)+\frac{\delta}3 f_{xx}(x) + O(\delta^3)  
\end{equation} 
as $\delta\downarrow 0$ uniformly for $x\in\R$ \cite[Appendix~A]{scoufis2005}. 
We scale 
\begin{equation} 
\label{eq:scaling1} 
\mU(x,t)\to \frac{\delta}{3}\mU\left(x,\frac{\delta}{3}t\right) 
\end{equation} 
in \eqref{eq:sILW} (and change variables accordingly) to obtain
\begin{equation} 
\mU_t + \{\mU,\mU_x\} +\frac{3}{\delta^2}\mU_x + \frac{3}{\delta}T\mU_{xx} + \ii[\mU,T\mU_x]=0. 
\end{equation} 
Expanding this in powers of $\delta$ using \eqref{eq:Texpand} yields 
\begin{equation} 
\mU_t+\{\mU,\mU_x\}+\mU_{xxx} +\frac{\delta}3\ii[\mU,\mU_x] + O(\delta^2)=0, 
\end{equation} 
which converges to the matrix KdV equation \eqref{eq:mKdV} in the limit $\delta\downarrow 0$. 

Similarly, by scaling 
\begin{equation} 
\label{eq:scaling2} 
\mU(x,t)\to \frac1\delta \mU\left(x,\frac{t}{3} \right), 
\end{equation} 
\eqref{eq:sILW} is transformed to
\begin{equation} 
\mU_t +\frac{3}{\delta}\{\mU,\mU_x\} + \frac{3}{\delta}\mU_x + 3T\mU_{xx}+\frac{3}{\delta}\ii[\mU,T\mU_x]=0. 
\end{equation} 
By inserting this expansion in \eqref{eq:Texpand}, we obtain
 \begin{equation} 
 \mU_t + \frac{3}{\delta}\{\mU,\mU_x\}+ \delta \mU_{xxx}+\ii[\mU,\mU_{xx}]+O(\delta^2)=0, 
 \end{equation}  
yielding the  generalized HF equation \eqref{eq:sHF} in the limit $\delta\downarrow 0$ provided $\{\mU,\mU_x\}=0$; the latter is achieved by imposing the condition $\mU^2=I$ (note that this condition is preserved under the time evolution \eqref{eq:sHF}).

\subsection{Spin generalization of the bidirectional Benjamin-Ono equation}
\label{sec:2sBO} 
The bidirectional BO equation was introduced by Abanov, Bettelheim, and Wiegmann \cite{abanov2009} as a powerful tool to derive a hydrodynamic description of the $A$-type  trigonometric CM system (to mention just one of several interesting applications). In this section, we present a spin generalization of this equation.

Let $\alpha(z)$ \eqref{eq:alpha} and $V(z)$ \eqref{eq:V} in II (trigonometric case).  
We assume that the variables $a_j\in\C$, $|e_j\rangle\in\cV$, $\langle f_j|\in\cV^*$ for $j=1,\ldots,N$ and $b_j\in\C$, $|g_j\rangle\in\cV$, $\langle h_j|\in\cV^*$ for $j=1,\ldots,N$ satisfying the first order equations  \eqref{eq:sCM1b}, \eqref{eq:sCM2b} and \eqref{eq:BT}, together with the constraints \eqref{eq:sCM1c} and  \eqref{eq:sCM2c}; note that these are exactly the equations defining the B\"acklund transformation of the $A$-type sCM system in the trigonometric case (see Proposition~\ref{prop:BT}). Using that, we define the functions
\begin{equation} 
\mU_0(z,t)\coloneqq  -\ii\sum_{j=1}^M |g_j(t)\rangle\langle h_j(t)|\alpha(z-b_j(t)), \quad \mU_1(z,t)\coloneqq  \ii\sum_{j=1}^N |e_j(t)\rangle\langle f_j(t)|\alpha(z-a_j(t)) 
\end{equation} 
of the position variable $z\in\C$ and time $t\in\R$, and 
\begin{equation} 
\mU\coloneqq \mU_0+\mU_1,\quad \tilde{\mU}\coloneqq \mU_0-\mU_1. 
\end{equation} 
One then can verify, by arguments similar to the ones we use to prove Propositions~\ref{prop:sBO1gen}, 
that the following equation is satisfied, 
\begin{equation} 
\label{eq:bsBO}
\mU_t + \{\mU,\mU_z\} +\ii\tilde{\mU}_{zz}-[\mU,\tilde{\mU}_z]=0. 
\end{equation} 
This is the spin generalization of the bidirectional BO equation. It would be interesting to generalize the arguments in  \cite{abanov2009} to derive a hydrodynamic description of the sCM systems from \eqref{eq:bsBO}. 
 
\appendix 
\section{Functional identities}
\label{app:alphaV}
Let $\alpha(z)$ and $V(z)$ be as in \eqref{eq:alpha} and \eqref{eq:V}, respectively. 
These functions obey various well-known functional identities which we use in the derivations of our results. 
For the convenience of the reader, we collect these identities here. 

First, the functions $\alpha(z)$ and $V(z)$ are  odd and even, respectively: 
\begin{equation} 
\label{eq:Id4} 
\alpha(-z)=-\alpha(z),\quad V(-z)=V(z)\quad (z\in\C).
\end{equation} 
Second, 
\begin{equation} 
\label{eq:IdV} 
V(z)=-\alpha'(z)= \alpha(z)^2+C ,
\end{equation} 
with different constants $C$ in the three cases I--III: 
\begin{equation} 
\label{eq:C} 
C\coloneqq \begin{cases} 0 & \text{(I)}\\ 
(\pi/L)^2 &  \text{(II)}\\
-(\pi/2\delta)^2 &  \text{(III)}.
\end{cases}
\end{equation} 
Third,
\begin{equation} 
\label{eq:Idmain} 
\alpha(a-b)\alpha(b-c) + \alpha(b-c)\alpha(c-a) + \alpha(c-a)\alpha(a-b) = C \quad (a,b,c \in \mathbb{C}),
\end{equation} 
with $C$ as in \eqref{eq:C}.

Finally, in case III, $\alpha(z)$ and $V(z)$ both are are both $2\ii\delta$-periodic: 
\begin{equation} \label{eq:Idperiodic}
\alpha(z+2\ii\delta)=\alpha(z),\quad V(z+2\ii\delta)=V(z)\quad (z\in\C)\quad \text{(III)}.
\end{equation} 

\section{Non-hermitian solutions of the sBO equation}
\label{app:sBOsolutions}
In the main text, we presented hermitian solutions of the sBO equation which, as we believe, are the most interesting from a physics point of view. 
We now present more general non-hermitian solutions which are interesting from a mathematics point of view. 
It is easy to adapt the arguments in the main text to prove this more general result.

\subsection{Result}
It is easy to check that the following simplifies to Theorem~\ref{thm:sBO} under the reduction \eqref{eq:reduction}; the latter is equivalent to $\mU=\mU^\dag$.

\begin{theorem} 
\label{thm:sBOgen} 
Let $\{a_j(t),|e_j(t)\rangle,\langle f_j(t)|\}_{j=1}^N$ and $\{b_j(t),|g_j(t)\rangle,\langle h_j(t)|\}_{j=1}^M$ be a solution of the time evolution equations \eqref{eq:sCM1a}--\eqref{eq:sCM1b}
and  \eqref{eq:sCM2a}--\eqref{eq:sCM2b}, respectively, with initial conditions 
\begin{equation} 
\begin{split} 
a_j(0)=a_{j,0},\quad \dot a_j(0)=v_j,\quad |e_j(0)\rangle =  |e_{j,0}\rangle,\quad \langle f_j(0)|=\langle f_{j,0}|\quad &(j=1,\ldots,N),\\
b_j(0)=b_{j,0},\quad \dot b_j(0)=w_j,\quad |g_j(0)\rangle =  |g_{j,0}\rangle,\quad \langle h_j(0)|=\langle h_{j,0}|\quad &(j=1,\ldots,M), 
\end{split} 
\end{equation} 
satisfying the constraints 
\begin{equation} 
\im(a_{j,0})<0 \quad (j=1,\ldots,N), \quad \im(b_{j,0})>0 \quad (j=1,\ldots,M), 
\end{equation} 
\begin{equation} 
\langle f_{j,0}|e_{j,0}\rangle =1 \quad (j=1,\ldots,N),\quad  \langle g_{j,0}|h_{j,0}\rangle =1 \quad (j=1,\ldots,M),
\end{equation} 
and 
\begin{subequations} 
\begin{align} 
\label{eq:BTat=0} 
\begin{split} 
\dot a_{j,0} \langle f_{j,0}| = & 2\ii\sum_{k\neq j}^N  \langle f_{j,0}|e_{k,0}\rangle \langle f_{k,0}|\alpha(a_{j,0}-a_{k,0}) \\
&-2\ii\sum_{k=1}^M  \langle f_{j,0}|g_{k,0}\rangle\langle h_{k,0}|\alpha(a_{j,0}-b_{k,0})
\end{split} \quad (j=1,\ldots,N), \\
\label{eq:BTbt=0} 
\begin{split} 
\dot b_{j,0} |g_{j,0}\rangle =& -2\ii\sum_{k\neq j}^M |g_{k,0}\rangle\langle h_{k,0}|g_{j,0}\rangle\alpha(b_{j,0}-b_{k,0})\\
& +2\ii\sum_{k=1}^N|e_{k,0}\rangle\langle f_{k,0}|g_{j,0}\rangle\alpha(b_{j,0}-a_{k,0})
\end{split} 
\quad (j=1,\ldots,M).   
\end{align} 
\end{subequations} 
Then 
\begin{equation} 
\label{eq:ansatzgen} 
\mU(x,t)= \ii\sum_{j=1}^N |e_j(t)\rangle\langle f_j(t)|\alpha(x-a_j(t)) -  \ii\sum_{j=1}^M |g_j(t)\rangle\langle h_j(t)|\alpha(x-b_j(t)) 
\end{equation} 
is a solution of the sBO equation \eqref{eq:sBO} for all times $t$ provided that the following condition holds true,
\begin{equation}
\label{eq:imajimbj1} 
\im(a_j(t))<0 \quad (j=1,\ldots,N),\quad \im(b_j(t))>0 \quad (j=1,\ldots,M). 
\end{equation} 
\end{theorem} 

\subsection{Proof of Theorem~\ref{thm:sBOgen}}
We start with the spin-pole ansatz
\begin{equation}
\label{eq:sBOansatzgen}
\mU(x,t)=\ii\sum_{j=1}^N \mP_j(t)\alpha(x-a_j(t))-\ii\sum_{j=1}^M \mQ_j(t)\alpha(x-b_j(t)),
\end{equation}
where $\mP_j=\mP_j(t)$ and $\mQ_j=\mQ_j(t)$ are $d\times d$ matrices and $a_j=a_j(t)\in\C$ and $b_j=b_j(t)\in\C$ satisfy \eqref{eq:imajimbj1}. 
With the notation
\begin{equation}
\label{eq:shorthandgen}
(a_j,\mP_j,r_j)=\begin{cases}
(a_j,\mP_j,+1) & j=1,\ldots,N \\
(b_{j-N},\mQ_{j-N},-1) & j=N+1,\ldots,\cN
\end{cases}, \quad \cN=N+M,
\end{equation}
in generalization of \eqref{eq:shorthand}, \eqref{eq:sBOansatzgen} may be written in the form \eqref{eq:U}. Inserting \eqref{eq:U} with \eqref{eq:shorthandgen} into \eqref{eq:sBO}, the proof of Proposition~\ref{prop:sBO1} may then be followed line-by-line to give the following.

\begin{proposition} 
\label{prop:sBO1gen} 
The function $\mU(x,t)$ in \eqref{eq:sBOansatzgen} satisfies the sBO equation \eqref{eq:sBO} provided \eqref{eq:imajimbj1} and the following equations hold true,
\begin{subequations} 
\label{eq:BT2gen}
\begin{align} 
\label{eq:BT2at} 
\dot a_{j} \mP_{j} = +2\ii\sum_{k\neq j}^N \mP_{j} \mP_{k}\alpha(a_{j}-a_{k}) -2\ii\sum_{k=1}^M \mP_{j}\mQ_{k}\alpha(a_{j}-b_{k})\quad (j=1,\ldots,N), \\
\label{eq:BT2bt} 
\dot b_{j} \mQ_{j} = -2\ii \sum_{k\neq j}^M \mQ_{k}\mQ_{j}\alpha(b_{j}-b_{k})+2\ii\sum_{k=1}^N \mP_{k}\mQ_{j}\alpha(b_{j}-a_{k})\quad (j=1,\ldots,M), 
\end{align} 
\end{subequations}
\begin{subequations} 
\begin{align} 
\dot \mP_j & = -2\ii\sum_{k \neq j}^N[\mP_j,\mP_k]V(a_j-a_k) \quad (j=1,\ldots,N),\\
\dot \mQ_j & = -2\ii\sum_{k \neq j}^N[\mQ_j,\mQ_k]V(b_j-b_k) \quad (j=1,\ldots,M),
\end{align} 
\end{subequations} 
and 
\begin{equation} 
\label{eq:PiQk} 
\mP_j^2=\mP_j  \quad (j=1,\ldots, N),\quad \mQ_j^2=\mQ_j \quad  (j=1,\ldots, M).
\end{equation} 
\end{proposition} 

We now make the ansatz 
\begin{equation}
\label{eq:mPmQansatz}
\mP_j= |e_j\rangle\langle f_j| \quad (j=1,\ldots,N),\qquad \mQ_j=|g_j\rangle\langle h_j| \quad (j=1,\ldots,M),
\end{equation}
with vectors $|e_j\rangle,|g_j\rangle\in\cV$ and $\langle f_j|,\langle h_j|\in\cV^*$. Similarly as in the special case \eqref{eq:ansatzh} treated in Theorem~\ref{thm:sBO}, $\mU(x,t)$ \eqref{eq:sBOansatzgen} satisfies the equations in Proposition~\eqref{prop:sBO1gen} if $\{a_j,|e_j\rangle,\langle f_j|\}_{j=1}^N$ and $\{b_j,|g_j\rangle,\langle h_j|\}_{j=1}^M$ satisfy the equations defining the B\"{a}cklund transformation of the sCM system discussed in Section~\ref{sec:sCM_BT}.

\begin{lemma}
Suppose that $\{a_j,|e_j\rangle,\langle f_j|\}_{j=1}^N$ and $\{b_j,|g_j\rangle,\langle h_j|\}_{j=1}^N$ satisfy the equations \eqref{eq:sCM1b}, \eqref{eq:sCM1c}, \eqref{eq:sCM2b}, \eqref{eq:sCM2c}, and \eqref{eq:BT} and $\mP_j$ and $\mQ_j$ are given by \eqref{eq:mPmQansatz}. Then, the function $\mU$ defined in \eqref{eq:sBOansatzgen} satisfies \eqref{eq:BT2gen}--\eqref{eq:PiQk}. 
\end{lemma}

\begin{proof}
By multiplying \eqref{eq:BTa} from the left by $|e_j\rangle$, one gets
\begin{subequations}
\label{eq:BT2genalt}
\begin{equation}
\dot{a}_j|e_j\rangle\langle f_j| =2\ii \sum_{k\neq j}^N |e_j\rangle \langle f_j | e_k \rangle \langle f_k| \alpha(a_j-a_k)-2\ii \sum_{k=1}^M |e_j\rangle \langle f_j | g_k \rangle \langle h_k| \alpha(a_j-b_k) \quad (j=1,\ldots,N). 
\end{equation} 
and, by multiplying \eqref{eq:BTb} from the right by $\langle h_j|$, one gets
\begin{equation}
\dot{b}_j|g_j\rangle\langle h_j| =-2\ii \sum_{k\neq j}^M | g_k \rangle\langle h_k | g_j \rangle \langle h_j|  \alpha(b_j-b_k)+2\ii \sum_{k=1}^N |e_k \rangle\langle f_k | g_j \rangle \langle h_j| \alpha(b_j-a_k) \quad (j=1,\ldots,M).
\end{equation} 
\end{subequations}
Then, \eqref{eq:BT2genalt} matches \eqref{eq:BT2gen}.
The remainder of the proof of the Lemma is similar to that of Lemma~\ref{lem:sBOlemma} and hence omitted. 
\end{proof}

Then, the theorem is implied by Proposition~\ref{prop:BT}.

\section{Non-hermitian solutions of the sncILW equation}
\label{app:sncILWsolutions}
In this appendix, we present non-hermitian solutions of the sncILW equation analogous to those of the sBO equation presented in Appendix~\ref{app:sBOsolutions}, together with a self-contained proof. 

\subsection{Result}
It is easy to check that the following simplifies to Theorem~\ref{thm:sncILW} under the reduction \eqref{eq:reduction}; the latter is equivalent to $\mU=\mU^\dag$. In the main text, we only specified the key ingredients in the proof of Theorem~\ref{thm:sncILW}. Here, we prove the following more general result using the notation and methods discussed in Section~\ref{sec:sncILW}. The proof structure parallels that of Theorem~\ref{thm:sBO} (and of Theorem~\ref{thm:sBOgen}, see Appendix~\ref{app:sBOsolutions}) presented in Section~\ref{sec:sBOproof}.
\begin{theorem} 
\label{thm:sncILWgen} 
Let $\{a_j(t),|e_j(t)\rangle,\langle f_j(t)|\}_{j=1}^N$ and $\{b_j(t),|g_j(t)\rangle,\langle h_j(t)|\}_{j=1}^M$ be a solution of the time evolution equations \eqref{eq:sCM1a}--\eqref{eq:sCM1b}
and  \eqref{eq:sCM2a}--\eqref{eq:sCM2b}, respectively, with initial conditions 
\begin{equation} 
\begin{split} 
a_j(0)=a_{j,0},\quad \dot a_j(0)=v_j,\quad |e_j(0)\rangle =  |e_{j,0}\rangle,\quad \langle f_j(0)|=\langle f_{j,0}|\quad &(j=1,\ldots,N),\\
b_j(0)=b_{j,0},\quad \dot b_j(0)=w_j,\quad |g_j(0)\rangle =  |g_{j,0}\rangle,\quad \langle h_j(0)|=\langle h_{j,0}|\quad &(j=1,\ldots,M), 
\end{split} 
\end{equation} 
satisfying the constraints 
\begin{equation} 
-\frac{3\delta}{2}<\im(a_{j,0})<-\frac{\delta}{2} \quad (j=1,\ldots,N), \quad \frac{\delta}{2}<\im(b_{j,0})<\frac{3\delta}{2} \quad (j=1,\ldots,M), 
\end{equation} 
\begin{equation} 
\langle f_{j,0}|e_{j,0}\rangle =1 \quad (j=1,\ldots,N),\quad  \langle g_{j,0}|h_{j,0}\rangle =1 \quad (j=1,\ldots,M),
\end{equation} 
and 
\begin{subequations} 
\begin{align} 
\label{eq:BTat=02_sncILW} 
\begin{split} 
\dot a_{j,0} \langle f_{j,0}| = & 2\ii\sum_{k\neq j}^N  \langle f_{j,0}|e_{k,0}\rangle \langle f_{k,0}|\alpha(a_{j,0}-a_{k,0}) \\
&-2\ii\sum_{k=1}^M  \langle f_{j,0}|g_{k,0}\rangle\langle h_{k,0}|\alpha(a_{j,0}-b_{k,0}+\ii\delta), 
\end{split} \quad (j=1,\ldots,N), \\
\label{eq:BTbt=02_sncILW} 
\begin{split} 
\dot b_{j,0} |g_{j,0}\rangle =& -2\ii\sum_{k\neq j}^M |g_{k,0}\rangle\langle h_{k,0}|g_{j,0}\rangle\alpha(b_{j,0}-b_{k,0})\\
& +2\ii\sum_{k=1}^N|e_{k,0}\rangle\langle f_{k,0}|g_{j,0}\rangle\alpha(b_{j,0}-a_{k,0}+\ii\delta),    
\end{split} 
\quad (j=1,\ldots,M).   
\end{align} 
\end{subequations} 
Then,
\begin{equation} 
\label{eq:ansatz2_ncILW} 
\begin{split}
\mU(x,t)=&\;  \ii\sum_{j=1}^N |e_j(t)\rangle\langle f_j(t)|\alpha(x-a_j(t)+\ii\delta/2) -  \ii\sum_{j=1}^M |g_j(t)\rangle\langle h_j(t)|\alpha(x-b_j(t)-\ii\delta/2),  \\
\mV(x,t)=&\; - \ii\sum_{j=1}^N |e_j(t)\rangle\langle f_j(t)|\alpha(x-a_j(t)-\ii\delta/2) +  \ii\sum_{j=1}^M |g_j(t)\rangle\langle h_j(t)|\alpha(x-b_j(t)+\ii\delta/2) 
\end{split}
\end{equation} 
is a solution of the sncILW equation \eqref{eq:sncILW} for all times $t$ provided that the following condition holds true,
\begin{equation}
\label{eq:imajimbj_sncILW}
-\frac{3\delta}{2}<\im(a_j(t))<-\frac{\delta}{2} \quad (j=1,\ldots,N),\quad \frac{\delta}{2}<\im(b_j(t))<\frac{3\delta}{2} \quad (j=1,\ldots,M). 
\end{equation} 
\end{theorem} 

\subsection{Proof of Theorem~\ref{thm:sncILWgen}}

We start with the spin-pole ansatz
\begin{equation}
\label{eq:sncILWansatzgen}
\cU(x,t)=\ii\sum_{j=1}^N \mP_j(t)\cA_{+}(x-a_j(t))-\ii\sum_{j=1}^M \mQ_j(t)\cA_{-}(x-b_j(t)),
\end{equation}
with $\cA_{\pm}$ as in \eqref{eq:cA} and where $\mP_j=\mP_j(t)$ and $\mQ_j=\mQ_j(t)$ are $d\times d$ matrices and $a_j=a_j(t)\in \C$, and $b_j=b_j(t)\in \C$ satisfy \eqref{eq:imajimbj_sncILW}   . Using the notation \eqref{eq:shorthandgen}, \eqref{eq:sncILWansatzgen} may be written in the form \eqref{eq:ansatz_sncILW}. By inserting \eqref{eq:ansatz_sncILW} with \eqref{eq:shorthandgen} into the sncILW equation \eqref{eq:sncILW2} and performing some computations, we obtain the following result. 

\begin{proposition} 
\label{prop:sncILW1gen} 
The functions $\mU(x,t)$ and $\mV(x,t)$ in \eqref{eq:ansatz2_ncILW} satisfy the sncILW equation \eqref{eq:sncILW} provided the following equations hold true,
\begin{subequations} \label{eq:BT2gen_sncILW}
\begin{align} 
\label{eq:BT2at2_sncILW} 
\dot a_{j} \mP_{j} = +2\ii\sum_{k\neq j}^N \mP_{j} \mP_{k}\alpha(a_{j}-a_{k}) -2\ii\sum_{k=1}^M \mP_{j}\mQ_{k}\alpha(a_{j}-b_{k}+\ii\delta)\quad (j=1,\ldots,N), \\
\label{eq:BT2bt2_sncILW} 
\dot b_{j} \mQ_{j} = -2\ii \sum_{k\neq j}^M \mQ_{k}\mQ_{j}\alpha(b_{j}-b_{k})+2\ii\sum_{k=1}^N \mP_{k}\mQ_{j}\alpha(b_{j}-a_{k}+\ii\delta)\quad (j=1,\ldots,M), 
\end{align} 
\end{subequations}
\begin{subequations} \label{eq:mPdotmQdot_sncILW}
\begin{align} 
\dot \mP_j & = -2\ii\sum_{k \neq j}^N[\mP_j,\mP_k]V(a_j-a_k) \quad (j=1,\ldots,N),\\
\dot \mQ_j & = -2\ii\sum_{k \neq j}^M[\mQ_j,\mQ_k]V(b_j-b_k) \quad (j=1,\ldots,M),
\end{align} 
\end{subequations} 
and 
\begin{equation} 
\label{eq:PiQk2_sncILW} 
\mP_j^2=\mP_j  \quad (j=1,\ldots, N),\quad \mQ_j^2=\mQ_j \quad  (j=1,\ldots, M).
\end{equation} 
\end{proposition} 

\begin{proof}
We compute each term in \eqref{eq:sncILW2}. Using the notation \eqref{eq:shorthandgen}, we start with
\begin{equation}
\label{eq:Ut_sncILW}
\cU_t=  \sum_{j=1}^{\cN} \bigg( \ii r_j \dot{\mP}_j  \mathcal{A}_{r_j}(x-a_j) - \ii r_j \mP_j \dot{a}_j \cA_{r_j}'(x-a_j)\bigg).
\end{equation}
Next, we compute 
\begin{align}
\label{eq:UUx_sncILW}
 \{\cU\ocomma \cU_x\}=&\;   -\sum_{j=1}^{\mathcal{N}}\sum_{k=1}^{\mathcal{N}} r_jr_k \{\mP_j,\mP_k\}  \cA_{r_j}(x-a_j)\circ\cA_{r_k}'(x-a_k)\nonumber \\
=&\; -2\sum_{j=1}^{\cN} \mP_j^2\cA_{r_j}(x-a_j)\circ\cA_{r_j}'(x-a_j)-\sum_{j=1}^{\mathcal{N}}\sum_{k\neq j}^{\mathcal{N}} r_jr_k \{\mP_j,\mP_k\}  \cA_{r_j}(x-a_j)\circ\cA_{r_k}'(x-a_k) \nonumber \\
=&\; \sum_{j=1}^{\cN} \mP_j^2\cA_{r_j}''(x-a_j)+ \sum_{j=1}^{\cN}\sum_{k\neq j}^{\cN} r_jr_k \{\mP_j,\mP_k\} \alpha(a_j-a_k+\ii(r_j-r_k)\delta/2)\cA_{r_k}'(x-a_k) \nonumber \\
&\; +\sum_{j=1}^{\cN}\sum_{k\neq j}^{\cN} r_j r_k \{\mP_j,\mP_k\} V(a_j-a_k+\ii(r_j-r_k)\delta/2)\big(\cA_{r_j}(x-a_j)-\cA_{r_k}(x-a_k)\big) ,
\end{align}
where we have used the identities \eqref{eq:Aidentity1} and \eqref{eq:Aidentity2}. 

The final sum in \eqref{eq:UUx_sncILW} vanishes because the summand is antisymmetric under the interchange $j\leftrightarrow k$. Hence, after re-labelling summation indices $j\leftrightarrow k$ in the second sum in \eqref{eq:UUx_sncILW} using \eqref{eq:Id4}, we are left with
\begin{equation}
\label{eq:UUx2_sncILW}
\{\cU\ocomma\cU_x\}= \sum_{j=1}^{\cN}  \mP_j^2 \cA_{r_j}''(x-a_j)- \sum_{j=1}^{\cN}\sum_{k\neq j}^{\cN} r_jr_k \{\mP_j,\mP_k\}\alpha(a_j-a_k+\ii(r_j-r_k)\delta/2)\cA_{r_j}'(x-a_j).
\end{equation}

To compute the terms in \eqref{eq:sncILW2} involving $\cT$, we use \eqref{eq:Aidentity3}. Applied to \eqref{eq:ansatz_sncILW}, this gives
\begin{equation}
\label{eq:TUx}
\cT \cU_x= -\sum_{j=1}^{\cN} \mP_j \cA_{r_j}'(x-a_j) ,\qquad \cT \cU_{xx}= -\sum_{j=1}^{\cN} \mP_j \cA_{r_j}''(x-a_j),
\end{equation}
where we have used the fact that $\cT$ commutes with differentiation to obtain the second equation from the first. 

From \eqref{eq:ansatz_sncILW} and the first equation in \eqref{eq:TUx}, we compute
\begin{align}
\label{eq:UTUx}
\ii  [\cU\ocomma\cT \cU_x] =&\;  \sum_{j=1}^{\cN}\sum_{k\neq j}^{\cN} r_j [\mP_j,\mP_k]  \cA_{r_j}(x-a_j)\circ \cA_{r_k}'(x-a_k) \nonumber\\
=&\; -\sum_{j=1}^{\cN}\sum_{k\neq j}^{\cN} r_j[\mP_j,\mP_k] \alpha(a_j-a_k+\ii(r_j-r_k)\delta/2)\cA_{r_k}'(x-a_k) \nonumber \\
&\; - \sum_{j=1}^{\cN}\sum_{k\neq j}^{\cN} r_j[\mP_j,\mP_k] V(a_j-a_k+\ii(r_j-r_k)\delta/2)\big(\cA_{r_j}(x-a_j)-\cA_{r_k}(x-a_k)\big),
\end{align}
inserting \eqref{eq:Aidentity2} in the second step. We can rewrite the second sum as follows:
\begin{multline}
\sum_{j=1}^{\cN}\sum_{k\neq j}^{\cN} r_j[\mP_j,\mP_k] V(a_j-a_k+\ii(r_j-r_k)\delta/2)\big(\cA_{r_j}(x-a_j)-\cA_{r_k}(x-a_k)\big) \\
=   \frac12 \sum_{j=1}^{\cN}\sum_{k\neq j}^{\cN} (r_j+r_k)[\mP_j,\mP_k] V(a_j-a_k+\ii(r_j-r_k)\delta/2)\big(\cA_{r_j}(x-a_j)-\cA_{r_k}(x-a_k)\big) \\
=    \sum_{j=1}^{\cN}\sum_{k\neq j}^{\cN} (r_j+r_k)[\mP_j,\mP_k] V(a_j-a_k+\ii(r_j-r_k)\delta/2)\cA_{r_j}(x-a_j),
\end{multline}
since $V(z)$ is an even function. Also changing variables $j\leftrightarrow k$ in the first sum in \eqref{eq:UTUx} using that $\alpha(z)$ is odd and $2\ii\delta$-periodic \eqref{eq:Idperiodic},
 we arrive at
\begin{align}
\label{eq:UTUx2}
\ii  [\cU\ocomma\cT \cU_x] =&\; -\sum_{j=1}^{\cN}\sum_{k\neq j}^{\cN} r_k[\mP_j,\mP_k] \alpha(a_j-a_k+\ii(r_j-r_k)\delta/2)\cA_{r_j}'(x-a_j) \nonumber \\
&\; - \sum_{j=1}^{\cN}\sum_{k \neq j}^{\cN} (r_j+r_k) [\mP_j,\mP_k]V(a_j-a_k+\ii(r_j-r_k)\delta/2)\cA_{r_j}(x-a_j). 
\end{align}
Inserting \eqref{eq:Ut_sncILW}, \eqref{eq:UUx2_sncILW}, the second equation in \eqref{eq:TUx}, and \eqref{eq:UTUx2} into \eqref{eq:sncILW2} yields
\begin{align}
\label{eq:auxsystem2_hyperbolic}
0=&\; \sum_{j=1}^{\cN} \Bigg( \ii r_j \dot{\mP}_j -   \sum_{k\neq j}^{\cN} (r_j+r_k)[\mP_j,\mP_k]V(a_j-a_k+\ii(r_j-r_k)\delta/2)\Bigg)\cA_{r_j}(x-a_j) \nonumber \\
&\; + \sum_{j=1}^{\cN} \Bigg( -\ii r_j\mP_j\dot{a}_j -\sum_{k\neq j}^{\cN} r_k\big(r_j \{\mP_j,\mP_k\}+[\mP_j,\mP_k]  \big)  \alpha(a_j-a_k+\ii(r_j-r_k)\delta/2)   \Bigg) \cA_{r_j}'(x-a_j) \nonumber \\
&\; + \sum_{j=1}^{\cN}  \big( \mP_j^2-\mP_j   \big)    \cA_{r_j}''(x-a_j).
\end{align}
Thus, $\cU$ satisfies the sncILW equation \eqref{eq:sncILW2} if and only if the following conditions are fulfilled:
\begin{align}
\mP_j \dot{a}_j=&\; \ii\sum_{k\neq j}^{\cN} r_k \big(\{\mP_j,\mP_k\}+r_j[\mP_j,\mP_k]\big)\alpha(a_j-a_k+\ii(r_j-r_k)\delta/2), \label{eq:BT1_sncILW}\\
\dot{\mP}_j=&\; - \ii \sum_{k\neq j}^{\cN}(1+r_jr_k)[\mP_j,\mP_k]V(a_j-a_k+\ii(r_j-r_k)\delta/2), \label{eq:BT2_sncILW}\\
\mP_j^2= &\; \mP_j, \label{eq:BT3_sncILW}
\end{align}
for $j=1,\ldots,\cN$. Recalling \eqref{eq:shorthandgen}, one can check that \eqref{eq:BT1_sncILW}, \eqref{eq:BT2_sncILW}, and \eqref{eq:BT3_sncILW} are equivalent to \eqref{eq:BT2gen_sncILW}, \eqref{eq:mPdotmQdot_sncILW}, and \eqref{eq:PiQk2_sncILW} respectively. 
\end{proof}

We now make the ansatz 
\begin{equation}
\label{eq:mPj_sncILW} 
\begin{split}
\mP_j=   |e_j\rangle\langle f_j| \quad (j=1,\ldots,N),  \qquad \mQ_j = |g_j\rangle\langle h_j| \quad (j=1,\ldots,M),
\end{split}
\end{equation} 
with vectors $|e_j\rangle\in\cV$ and $\langle f_j|\in\cV^*$. 
Similar to the sBO case, the functions $\mU(x,t)$, $\mV(x,t)$ defined in \eqref{eq:ansatz2_ncILW} satisfy the equations in Proposition~\ref{prop:sncILW1gen} if $\{a_j,|e_j\rangle,\langle f_j|\}_{j=1}^N$ and $\{b_j,|e_j\rangle,\langle f_j|\}_{j=1}^M$ satisfy the equations defining the B\"acklund transformations of the sCM system discussed in Section~\ref{sec:sCM_BT}.  

\begin{lemma}
Let $\tilde{a}_j\coloneqq a_j-\ii\delta/2$ for $j=1,\ldots,N$ and $\tilde{b}_j=b_j+\ii\delta/2$ for $j=1,\ldots,M$. Suppose that $\{\tilde{a}_j,|e_j\rangle,\langle f_j|\}_{j=1}^N$ and $\{\tilde{b}_j,|g_j\rangle,\langle h_j|\}_{j=1}^M$ satisfy the equations \eqref{eq:sCM1b}, \eqref{eq:sCM1c}, \eqref{eq:sCM2b}, \eqref{eq:sCM2c}, and \eqref{eq:BT} and $a_j$, $b_j$, $\mP_j$ and $\mQ_j$ are given by \eqref{eq:mPj_sncILW}. Then \eqref{eq:BT2gen_sncILW}--\eqref{eq:PiQk2_sncILW} hold.
\end{lemma}

\begin{proof}
By multiplying \eqref{eq:BTa} with $a_j$ replaced by $\tilde{a}_j$ from the left by $|e_j\rangle$, one gets
\begin{subequations}
\label{eq:BT2genalt_sncILW}
\begin{equation}
\dot{\tilde{a}}_j|e_j\rangle\langle f_j| =2\ii \sum_{k\neq j}^N |e_j\rangle \langle f_j | e_k \rangle \langle f_k| \alpha(\tilde{a}_j-\tilde{a}_k)-2\ii \sum_{k=1}^M |e_j\rangle \langle f_j | g_k \rangle \langle h_k| \alpha(\tilde{a}_j-\tilde{b}_k) \quad (j=1,\ldots,N). 
\end{equation} 
Similarly, by multiplying \eqref{eq:BTb} with $b_j$ replaced by $\tilde{b}_j$ from the right by $\langle h_j|$, one gets
\begin{equation}
\dot{\tilde{b}}_j|g_j\rangle\langle h_j| =-2\ii \sum_{k\neq j}^M | g_k \rangle\langle h_k | g_j \rangle \langle h_j|  \alpha(\tilde{b}_j-\tilde{b}_k)+2\ii \sum_{k=1}^N |e_k \rangle\langle f_k | g_j \rangle \langle h_j| \alpha(\tilde{b}_j-\tilde{a}_k) \quad (j=1,\ldots,M).
\end{equation} 
\end{subequations}
Equations \eqref{eq:BT2genalt_sncILW} are seen to be equivalent to \eqref{eq:BT2gen_sncILW} after recalling \eqref{eq:mPj_sncILW}, the definitions of $\tilde{a}_j$ and $\tilde{b}_j$, and the $2\ii\delta$-periodicity \eqref{eq:Idperiodic} of $\alpha$.

The remainder of the proof is similar to that of Lemma~\ref{lem:sBOlemma} and hence omitted. 
\end{proof}

Then, the theorem is implied by Proposition~\ref{prop:BT}.

\noindent
{\bf Acknowledgement} {\it We would like to thank Rob Klabbers, Enno Lenzmann, Masatoshi Noumi and Junichi Shiraishi for helpful discussions.
The work of B.K.B. was supported by the Olle Engkvist Byggm\"{a}stare Foundation, Grant 211-0122. 
E.L. gratefully acknowledges  support from the European Research Council, Grant Agreement No.\ 2020-810451.
J.L. acknowledges support from the Ruth and Nils-Erik Stenb\"ack Foundation, the Swedish Research Council, Grant No.\ 2015-05430 and Grant No.\ 2021-03877, and the European Research Council, Grant Agreement No. 682537.
}


\begin{thebibliography}{10}

\bibitem{benjamin1967}
T.B. Benjamin.
\newblock Internal waves of permanent form in fluids of great depth.
\newblock {\em J. Fluid Mech.}, 29:559, 1967.

\bibitem{ono1975}
H.~Ono.
\newblock {Algebraic Solitary Waves in Stratified Fluids}.
\newblock {\em J. Phys. Soc. Japan}, 39:1082, 1975.

\bibitem{berntson2020}
B.K. Berntson, E.~Langmann, and J.~Lenells.
\newblock {Non-chiral Intermediate Long Wave equation and inter-edge effects in
  narrow quantum Hall systems}.
\newblock {\em Phys. Rev. B}, 102:155308, 2020.

\bibitem{berntsonlangmann2020}
B.K. Berntson, E.~Langmann, and J.~Lenells.
\newblock On the non-chiral intermediate long wave equation.
\newblock {\em Nonlinearity (in print)}, 2022.
\newblock arXiv preprint: nlin.SI/2005.10781.

\bibitem{berntsonlangmann2021}
B.K. Berntson, E.~Langmann, and J.~Lenells.
\newblock On the non-chiral intermediate long wave equation {II}: periodic
  case.
\newblock {\em Nonlinearity (in print)}, 2022.
\newblock arXiv preprint: nlin.SI/2103.02572.

\bibitem{zhou2015}
T.~Zhou and M.~Stone.
\newblock {Solitons in a continuous classical Haldane--Shastry spin chain}.
\newblock {\em Phys. Lett. A}, 379:2817, 2015.

\bibitem{lenzmann2018}
E.~Lenzmann and A.~Schikorra.
\newblock {On energy-critical half-wave maps into $\mathbb{S}^2$}.
\newblock {\em Invent. Math.}, 213:1, 2018.

\bibitem{lenzmann2018b}
E.~Lenzmann.
\newblock {A short primer on the Half-Wave Maps Equation}.
\newblock {\em J.\'{E}.D.P.}, 2018.

\bibitem{senthil1999}
T.~Senthil, J.B. Marston, and M.P.A. Fisher.
\newblock {Spin quantum Hall effect in unconventional superconductors}.
\newblock {\em Phys. Rev. B}, 60:4245, 1999.

\bibitem{mccann2006}
E.~McCann and V.I. Fal'ko.
\newblock {Landau-Level Degeneracy and Quantum Hall Effect in a Graphite
  Bilayer}.
\newblock {\em Phys. Rev. Lett.}, 96:086805, 2006.

\bibitem{bernevig2006}
B.A. Bernevig and S.-C. Zhang.
\newblock {Quantum Spin Hall Effect}.
\newblock {\em Phys. Rev. Lett.}, 96:106802, 2006.

\bibitem{gibbons1984}
J.~Gibbons and T.~Hermsen.
\newblock {A generalisation of the Calogero-Moser system}.
\newblock {\em Physica D}, 11:337, 1984.

\bibitem{wojciechowski1987}
S.~Wojciechowski.
\newblock {An integrable marriage of the Euler equations to the Calogero-Moser
  system}.
\newblock {\em Physica D}, 28:245, 1987.

\bibitem{olshanetsky1981}
M.A. Olshanetsky and A.M. Perelomov.
\newblock {Classical integrable finite-dimensional systems related to Lie
  algebras}.
\newblock {\em Phys. Rep.}, 71:313, 1981.

\bibitem{chen1979}
H.H. Chen, Y.C. Lee, and N.R. Pereira.
\newblock Algebraic internal wave solitons and the integrable
  {Calogero--Moser--Sutherland} {$N$}‐body problem.
\newblock {\em Phys. Fluids}, 22:187, 1979.

\bibitem{berntsonklabbers2020}
B.K. Berntson, R.~Klabbers, and E.~Langmann.
\newblock Multi-solitons of the half-wave maps equation and spin-pole
  {Calogero-Moser} dynamics.
\newblock {\em J. Phys. A: Math. Theor.}, 53:505702, 2020.

\bibitem{matsuno2022}
Y.~Matsuno.
\newblock Integrability, conservation laws and solitons of a many-body
  dynamical system associated with the half-wave maps equation.
\newblock {\em Physica D}, 430:133080, 2022.

\bibitem{abanov2009}
A.G. Abanov, E.~Bettelheim, and P.~Wiegmann.
\newblock {Integrable hydrodynamics of Calogero-Sutherland model: bidirectional
  Benjamin-Ono equation}.
\newblock {\em J. Phys. A: Math. Theor.}, 42:135201, 2009.

\bibitem{gibbons1983}
J.~Gibbons, T.~Hermsen, and S.~Wojciechowski.
\newblock A {B}{\"a}cklund transformation for a generalised {C}alogero-{M}oser
  system.
\newblock {\em Phys. Lett. A}, 94:251, 1983.

\bibitem{hartnoll2007}
S.A. Hartnoll, P.K. Kovtun, M.~M\"uller, and S.~Sachdev.
\newblock Theory of the {N}ernst effect near quantum phase transitions in
  condensed matter and in dyonic black holes.
\newblock {\em Phys. Rev. B}, 76:144502, 2007.

\bibitem{andreev2011}
A.V. Andreev, S.A. Kivelson, and B.~Spivak.
\newblock {Hydrodynamic Description of Transport in Strongly Correlated
  Electron Systems}.
\newblock {\em Phys. Rev. Lett.}, 106:256804, 2011.

\bibitem{svintsov2013}
D.~Svintsov, V.~Vyurkov, V.~Ryzhii, and T.~Otsuji.
\newblock Hydrodynamic electron transport and nonlinear waves in graphene.
\newblock {\em Phys. Rev. B}, 88:245444, 2013.

\bibitem{calogero1991}
F.~Calogero.
\newblock {\em {Why Are Certain Nonlinear PDEs Both Widely Applicable and
  Integrable?}}, pages 1--62.
\newblock Springer Berlin Heidelberg, Berlin, Heidelberg, 1991.

\bibitem{bettelheim2006}
E.~Bettelheim, A.G. Abanov, and P.~Wiegmann.
\newblock {Nonlinear Quantum Shock Waves in Fractional Quantum Hall Edge
  States}.
\newblock {\em Phys. Rev. Lett.}, 97:246401, 2006.

\bibitem{wiegmann2012}
P.~Wiegmann.
\newblock {Nonlinear Hydrodynamics and Fractionally Quantized Solitons at the
  Fractional Quantum {H}all Edge}.
\newblock {\em Phys. Rev. Lett.}, 108:206810, 2012.

\bibitem{abanov2005}
A.G. Abanov and P.B. Wiegmann.
\newblock {Quantum hydrodynamics, the Quantum Benjamin-Ono equation, and the
  Calogero model}.
\newblock {\em Phys. Rev. Lett.}, 95:076402, 2005.

\bibitem{difrancesco1997}
P.~Di~Francesco, P.~Mathieu, and D.~S\'en\'echal.
\newblock {\em {Conformal field theory}}.
\newblock Graduate texts in contemporary physics. Springer, New York, NY, 1997.

\bibitem{dirac1939}
P.A.M. Dirac.
\newblock A new notation for quantum mechanics.
\newblock {\em Math. Proc. Cambridge Philos. Soc.}, 35:416, 1939.

\bibitem{calogero1976coupled}
F.~Calogero and A.~Degasperis.
\newblock Coupled nonlinear evolution equations solvable via the inverse
  spectral transform, and solitons that come back: the boomeron.
\newblock {\em Lett. Nuovo Cim}, 16:425, 1976.

\bibitem{kodama1981}
Y.~Kodama, J.~Satsuma, and M.J. Ablowitz.
\newblock {Nonlinear Intermediate Long-Wave Equation: Analysis and Method of
  Solution}.
\newblock {\em Phys. Rev. Lett.}, 46:687, 1981.

\bibitem{DLMF}
{\it NIST Digital Library of Mathematical Functions}.
\newblock http://dlmf.nist.gov/, Release 1.0.26 of 2020-03-15.
\newblock F.W.J. Olver, A.B. {Olde Daalhuis}, D.W. Lozier, B.I. Schneider, R.F.
  Boisvert, C.W. Clark, B.R. Miller, B.V. Saunders, H.S. Cohl, and M.A.
  McClain, eds.

\bibitem{berntsonklabbers2021}
B.K. Berntson, R.~Klabbers, and E.~Langmann.
\newblock {The non-chiral intermediate Heisenberg ferromagnet equation}.
\newblock arXiv preprint: math-ph/2110.06239.

\bibitem{scoufis2005}
G.~Scoufis and C.M. Cosgrove.
\newblock An application of the inverse scattering transform to the modified
  intermediate long wave equation.
\newblock {\em J. Math. Phys.}, 46:103501, 2005.

\bibitem{lax1968}
P.D. Lax.
\newblock Integrals of nonlinear equations of evolution and solitary waves.
\newblock {\em Comm. Pure Appl. Math.}, 21:467, 1968.

\bibitem{goncharenko2001}
V.M. Goncharenko.
\newblock {Multisoliton Solutions of the Matrix KdV Equation}.
\newblock {\em Theor. Math. Phys.}, 126:81, 2001.

\bibitem{carey1999}
A.L. Carey and E.~Langmann.
\newblock {Loop Groups, Anyons and the Calogero--Sutherland Model}.
\newblock {\em Comm. Math. Phys.}, 201:1, 1999.

\bibitem{kulkarni2009}
M.~Kulkarni, F.~Franchini, and A.G. Abanov.
\newblock {Nonlinear dynamics of spin and charge in spin-Calogero model}.
\newblock {\em Phys. Rev. B}, 80:165105, 2009.

\bibitem{xing2015}
L.~Xing.
\newblock {\em {Classical hydrodynamics of Calogero-Sutherland models}}.
\newblock PhD thesis, University of Illinois at Urbana-Champaign, 2015.

\bibitem{hikami1993}
K.~Hikami and M.~Wadati.
\newblock Integrability of {C}alogero-{M}oser {S}pin {S}ystem.
\newblock {\em J. Phys. Soc. Japan}, 62:469, 1993.

\bibitem{krichever1995}
I.~Krichever, O.~Babelon, E.~Billey, and M.~Talon.
\newblock {Spin generalization of the Calogero-Moser system and the matrix KP
  equation}.
\newblock In S.P Novikov, editor, {\em Topics in Topology and Mathematical
  Physics}, volume 170, pages 83--120. American Mathematical Society, 1995.

\bibitem{king2009}
F.W. King.
\newblock {\em {Hilbert Transforms}}, volume~1 of {\em Encyclopedia of
  Mathematics and its Applications}.
\newblock Cambridge University Press, 2009.

\end{thebibliography}
\end{document}